\documentclass{article}
\usepackage{amsmath,amsthm,amssymb}
\usepackage{enumerate}
\usepackage[usenames,dvipsnames]{color}
\usepackage[bookmarks,colorlinks,breaklinks]{hyperref}  % PDF hyperlinks, with coloured links
\usepackage{doi,url}

\def\titlerunning#1{\gdef\titrun{#1}}
\makeatletter
\def\author#1{\gdef\autrun{\def\and{\unskip, }#1}\gdef\@author{#1}}
\def\address#1{{\def\and{\\\hspace*{18pt}}\renewcommand{\thefootnote}{}%
\footnote {#1}}%
\markboth{\autrun}{\titrun}}
\makeatother
\def\email#1{E-mail: #1}
\def\subjclass#1{\par\medskip
\noindent\textbf{Mathematics Subject Classification (2010).} #1}
\def\keywords#1{\par\medskip
\noindent\textbf{Keywords.} #1}

\definecolor{dullmagenta}{rgb}{0.4,0,0.4}   % #660066
\definecolor{darkblue}{rgb}{0,0,0.4}
\hypersetup{linkcolor=red,citecolor=blue,filecolor=dullmagenta,urlcolor=darkblue} 

\newcommand{\eq}[1]{\eqref{#1}}

\newtheorem{theorem}{Theorem}[section]
\newtheorem{proposition}[theorem]{Proposition}
\newtheorem{lemma}[theorem]{Lemma}
\newtheorem{corollary}[theorem]{Corollary}

\newtheorem{definition}[theorem]{Definition}
\newtheorem{remark}[theorem]{Remark}

\newtheorem*{induction}{Induction hypothesis}
\newtheorem*{remark*}{Remark}
\newtheorem*{remarks}{Remarks}

\numberwithin{equation}{section}

\DeclareMathOperator{\supp}{supp}
\DeclareMathOperator{\tr}{tr}
\DeclareMathOperator{\Ran}{Ran}
\DeclareMathOperator{\dist}{dist}

\DeclareMathOperator{\diam}{diam}

\newcommand{\pr}{\prime}

\newcommand\R{\mathbb R}
\newcommand\N{\mathbb N}

\newcommand\Z{\mathbb Z}
\newcommand\G{\mathbb{G}} 

\renewcommand\P{\mathbb P}
\newcommand\E{\mathbb E}

\renewcommand\L{\mathrm{L}}

\newcommand\bL{\boldsymbol{L}}

\newcommand\btau{\boldsymbol{\tau}}

\newcommand{\cD}{\mathcal{D}}
\newcommand{\cE}{\mathcal{E}}

\newcommand{\cG}{\mathcal{G}}

\newcommand{\cJ}{\mathcal{J}}

\newcommand{\cL}{\mathcal{L}}

\newcommand{\cS}{\mathcal{S}}

\newcommand\e{\mathrm{e}}

\newcommand\Chi{\raisebox{.2ex}{$\chi$}}

\newcommand{\abs}[1]{\left\lvert #1 \right\rvert}
\newcommand{\norm}[1]{\left\lVert #1 \right\rVert}
\newcommand{\scal}[1]{\left\langle #1 \right\rangle}
\newcommand{\set}[1]{\left\{ #1 \right\}}
\newcommand{\pa}[1]{\left( #1 \right)}
\newcommand{\br}[1]{\left[ #1 \right]}

\newcommand{\up}[1]{^{(#1)}}

\newcommand{\x}{\boldsymbol{x}}
\newcommand{\bolda}{\boldsymbol{a}}
\newcommand{\boldx}{\boldsymbol{x}}
\newcommand{\boldb}{\boldsymbol{b}}
\newcommand{\boldu}{\boldsymbol{u}}

\newcommand{\y}{\boldsymbol{y}}
\newcommand{\boldy}{\boldsymbol{y}}
\newcommand{\boldr}{\boldsymbol{r}}

\newcommand{\bolds}[1]{\boldsymbol{#1}}

\newcommand{\bom}{\boldsymbol{\omega}}

\newcommand{\blambda}{\boldsymbol{\Lambda}}
\newcommand{\blam}{\boldsymbol{\Lambda}}
\newcommand{\wlam}{\widetilde{\boldlambda}}

\newcommand{\wups}{\widetilde{\Upsilon}}
\newcommand{\ups}{\Upsilon}

\newcommand\eps{\varepsilon}

\newcommand\La{\Lambda}
\newcommand{\vphi}{\varphi}

\newcommand{\boxLx}{\Lambda_L(x)}

\newcommand{\boldlambda}{\mathbf{\Lambda}}

\newcommand{\NboxLx}{\mathbf{\Lambda}_{L}^{(N)}(\boldx)}

\newcommand{\NboxLu}{\mathbf{\Lambda}_{L}^{(N)}(\boldu)}

\newcommand{\Nboxlu}{\mathbf{\Lambda}_{\ell}^{(N)}(\boldu)}

\newcommand{\nboxLa}{\mathbf{\Lambda}_{L}^{(n)}(\bolda)}
\newcommand{\nboxLb}{\mathbf{\Lambda}_{L}^{(n)}(\boldb)}
\newcommand{\nboxLx}{\mathbf{\Lambda}_{L}^{(n)}(\boldx)}

\newcommand{\nboxLu}{\mathbf{\Lambda}_{L}^{(n)}(\boldu)}

\newcommand{\nboxx}{\mathbf{\Lambda}^{(n)}(\boldx)}
\newcommand{\nboxy}{\mathbf{\Lambda}^{(n)}(\boldy)}

\newcommand{\nboxlu}{\mathbf{\Lambda}_{\ell}^{(n)}(\boldu)}

\newcommand{\naboxLx}{ \boldlambda^{(n)}(\boldx) = \prod_{i = 1}^{n}  \,\, \Lambda_{L_i}(x_i)   }
\newcommand{\naboxLa}{ \boldlambda^{(n)}(\bolda) = \prod_{i = 1}^{n}  \,\, \Lambda_{L_i}(a_i)   }

\newcommand{\Rlamxy}{\Chi_{\x} R_{\blam}(E)\Chi_{\y}}
\newcommand{\Rlamab}{\Chi_{\bolda} R_{\blam}(E)\Chi_{\boldb}}

\newcommand{\Rlamxyz}{\Chi_{\x} R_{\blam}(z)\Chi_{\y}}

\newcommand{\Chix}{\Chi_{\x}}
\newcommand{\Chiy}{\Chi_{\y}}

\newcommand{\setn}{\bigl\{1, \, ... , \, n \bigr\}}

\newcommand{\Ndspace}{\mathbb{Z}^{Nd}}

\newcommand{\dspace}{\mathbb{Z}^{d}}

\newcommand{\wdel}{\widetilde{\delta}_{+}}

\newcommand{\Bl}{\Bigl}
\newcommand{\Br}{\Bigr}

\newcommand{\En}{E^{(n)}}
\newcommand{\EN}{E^{(N)}}

\newcommand\beq{\begin{equation}}
\newcommand\eeq{\end{equation}}

\newcommand{\qtx}[1]{\quad\text{#1}\quad}
\newcommand{\sqtx}[1]{\;\text{#1}\;}

\begin{document}

\titlerunning{Multi-particle continuous Anderson Hamiltonians}
\title{Bootstrap multiscale analysis  and localization for  multi-particle continuous Anderson Hamiltonians}

\author{Abel Klein\thanks{A.K. was  supported in part by the NSF under grant DMS-1301641.}
\and  
Son T. Nguyen}

\date{}
\maketitle

\address{A. Klein: {University of California, Irvine;
Department of Mathematics;
Irvine, CA 92697-3875,  USA.}
 \email{aklein@uci.edu}
\and
 S.T. Nguyen: {Mathematics Department;
University of Missouri;
Columbia, MO 65211 USA.}
  \email{sondgnguyen1@gmail.com}
}

\begin{abstract}
 We extend the bootstrap multiscale analysis developed by Germinet and Klein   to the multi-particle continuous Anderson Hamiltonian, obtaining  Anderson localization with finite multiplicity of eigenvalues,   decay of eigenfunction correlations, and a strong form of dynamical localization.  We do not require a covering condition.  The initial step for this multiscale analysis,   required to hold for   energies in a nontrivial interval at the bottom of the spectrum, is verified   for multi-particle continuous  Anderson Hamiltonians.  We also  extend the 
  unique continuation principle for spectral projections of Schr\" odinger operators  to arbitrary rectangles, and use it  to prove Wegner estimates for  multi-particle continuous Anderson Hamiltonians without the requirement of  a covering condition.

 \end{abstract}

\subjclass{Primary 82B44; Secondary 47B80, 60H25, 81Q10.}
\keywords{Multi-particle localization, multi-particle Anderson Hamiltonian, continuous Anderson Hamiltonian, multi-particle random Schr\"odinger operators, multiscale analysis, Anderson localization, dynamical localization, multi-particle Wegner estimates.}

 \tableofcontents

\section*{Introduction}
The multi-particle Anderson Hamiltonian is an  alloy-type   random 
Schr\"odinger operator describing  $n$ interacting electrons  moving in  a medium with random impurities. It is the continuous version of the (discrete) multi-particle Anderson model.

Localization was proved for  the multi-particle discrete  Anderson model   by Chulaevsky and Suhov \cite{CS1,CS2,CS3}, using a multiscale analysis,  and  Aizenman and Warzel \cite{AWmp}, using the fractional moment method. 
  Chulaevsky, Boutet de Monvel and Suhov \cite{CBS} extended the results of Chulaevsky and Suhov to the 
 multi-particle  continuous Anderson Hamiltonian, establishing Anderson and dynamical localization at the bottom of the spectrum.

 The bootstrap multiscale analysis,  developed in the one-particle case  by Germinet and Klein  \cite{GK1} (see also \cite{Kl}), is an enhanced multiscale analysis that yields  sub-exponentially decaying probabilities for `bad' events.  The  initial step  for the bootstrap  multiscale analysis only
requires the verification of polynomial decay of the finite volume 
resolvent, at some  sufficiently large  scale,
with probability bigger than some minimal probability   $1 -p_0  $, where $0< p_0 <1$ is independent of the scale. An important feature of the bootstrap multiscale analysis is that the final probability estimates are independent of the
probability estimate in  
 the initial step:  any desired sub-exponential decay for the
probabilities of  `bad' events can be achieved. The bootstrap multiscale analysis yields Anderson localization with finite multiplicity of eigenvalues,   decay of eigenfunction correlations, and a strong form of dynamical localization.

We  previously extended  the bootstrap multiscale analysis     to the multi-particle (discrete) Anderson model  \cite{KlN}.  The initial step for the bootstrap multiscale analysis of  \cite[Theorem~1.5]{KlN} has to hold for all energies in the spectrum (and hence for all energies); it can be verified  for the multi-particle Anderson model at high disorder, as discussed in \cite[Remark~1.6]{KlN}.

In this article we extend the bootstrap multiscale analysis (and its consequences)  to the multi-particle (continuous) Anderson Hamiltonian; we do not require a covering condition. The initial step is  only required to hold for  all energies in a nontrivial interval at the bottom of the spectrum (or equivalently, for all energies below some fixed energy).   We also show  that we always have  this initial step in some nontrivial interval at the bottom of the spectrum  for multi-particle  Anderson Hamiltonians.  The consequences to the bootstrap multiscale analysis include, in addition to Anderson and dynamical localization,  new results for multi-particle (continuous) Anderson Hamiltonians:
finite multiplicity of eigenvalues,   decay of eigenfunction correlations, and a strong form of dynamical localization (see Theorem~\ref{localizationH}).

Although the results in this paper are written for the continuous multi-particle  Anderson Hamiltonian, they also apply to the discrete multi-particle Anderson model, yielding localization at the bottom of the spectrum for the discrete model at any disorder.

The main definitions and results are stated in Section~\ref{secmain}. Theorem~\ref{localizationH} states that continuous multi-particle  Anderson Hamiltonians exhibit Anderson localization with finite multiplicity of eigenvalues,   decay of eigenfunction correlations, and a strong form of dynamical localization in an interval at the bottom of the spectrum.  
   Theorem~\ref{maintheorem} is the bootstrap multiscale analysis.  The consequences regarding  localization (Anderson localization with finite multiplicity of eigenvalues,   dynamical localization, decay of eigenfunction correlations)   are given in Corollary~\ref{localization}. In Section~\ref{secinitialMSA} we show that the hypotheses of Theorem~\ref{maintheorem} (the initial step for the bootstrap multiscale analysis) are always satisfied at some nontrivial  interval at the bottom of the spectrum. Section~\ref{sectoolkilt} contains a collection of technical results necessary for the multiscale analysis in the continuum.  The proof of Theorem~\ref{maintheorem}
is given in Section~\ref{secMSAproof}, and the derivation of Corollary~\ref{localization} is discussed in Section~\ref{secMSAloc}.

 In the multi-particle case events based on disjoint boxes are not necessarily independent, even if the boxes are far apart from each other.   This difficulty is overcome by  the use of the concepts of    partially and fully separated boxes (Subsection~\ref{secpfsep}) and  partially and fully interactive boxes (Subsection~\ref{secprint})  introduced by Chulaevsky and Suhov \cite{CS1,CS2,CS3}.  The relevant  distance between boxes is the Hausdorff distance (see \eq{Hausddist}), introduced in this context by  Aizenman and Warzel \cite{AWmp}. In the multiscale analysis partially interactive boxes are handled by the induction hypothesis, i.e., by the conclusions of Theorem \ref{maintheorem} for a smaller number of particles (see Lemmas~\ref{PIsuit} and \ref{PINSl}), and fully interactive boxes  are handled similarly to one particle boxes (see Lemma~\ref{part1prop2}).

The multiscale analysis requires Wegner estimates. Wegner estimates  were previously  proved for the $n$-particle discrete Anderson model 
 \cite{CS1,Ki2,KlN}.  In the continuum,  Wegner estimates for the  $n$-particle  Anderson Hamiltonian with a covering condition  were proved  in \cite{KZ,BCSS}, and without the covering condition in \cite{HK}.

The one-particle  energy interval multiscale analysis \cite{FMSS,vDK,GK1,Kl} requires a two-volume Wegner estimate, i.e.,  an estimate of the probability of the spectra of independent finite volume Hamiltonians   being close together.  Chulaevsky and Suhov \cite{CS1,CS2,CS3} realized that for $n$-particles this estimate is required for partially separated  finite volume Hamiltonians, that is, finite volume Hamiltonians on partially separated rectangles (here we need rectangles, not just boxes),
and proved such an estimate in the discrete case (see also \cite{KlN}).  In the continuum, such an estimate was proved  for the  $n$-particle  Anderson Hamiltonian  with a covering condition \cite{BCSS}. This two-volume Wegner estimate is now  proven  without the covering condition  in  \cite{HK} and in Corollary~\ref{Wegner2} below by somewhat different arguments. 

 Our definition of the  finite volume random potential (see \eq{rpn}), as well as our definition of fully and partially  separated rectangles (Definition~\ref{pfsep}), are slightly different than the ones used in \cite{BCSS,HK}. While \cite{BCSS,HK} take the finite volume random potential to be the restriction of the infinite volume random potential to  the $n$-particle rectangle,  our finite volume random potential  contains only random variables indexed by sites located in the faces of  the rectangle  (see \eq{rp1}).   We prove a Wegner estimate in Theorem~\ref{Wegner0} in which  the expectation is taken only with respect to the random variables indexed by one face of the rectangle (a  one-particle box).  In Corollary~\ref{Wegner2}  we derive from Theorem~\ref{Wegner0} a two-volume Wegner estimate  for partially separated rectangles as  in Definition~\ref{pfsep}. To do this, in Appendix~\ref{apUCP} we extend the results of  \cite{Kl2}, proving a unique continuation principle for spectral projections of Schr\" odinger operators on arbitrary rectangles.

\section{Main definitions and results}\label{secmain}

We start by defining the multi-particle  Anderson Hamiltonian.  We write   $\bolda = (a_1, \ldots, a_n) \in \R^{nd}\cong (\R^d)^n$, and  set
  $\norm{\bolda} := \max \{ \norm{a_1}, \ldots, \norm{a_n} \}$, where $\norm{x}=\norm{x}_\infty  := \max \{\abs{x_1}, \ldots, \abs{x_d} \}$ for $x = (x_1, \ldots, x_d) \in \R^{d}$.

\begin{definition}\label{defAndmodel}
Given $n\in \N$,  the $n$-particle Anderson Hamiltonian is  the  random Schr\"odinger 
operator on 
$L^{2}(\mathbb{R}^{nd})$ given by
\beq \label{AndH}
H_{\bom}^{(n)} :=H_{0,\bom}^{(n)} + U, \qtx{with} H_{0,\bom}^{(n)}: =  -\Delta^{(n)} + V_{\bom}^{(n)} ,
\eeq
where:
\begin{enumerate}
\item 
$\Delta^{(n)}$ is the $nd$-dimensional  Laplacian operator.

\item 
$V_{\bom}^{(n)}$ is the random potential  given by ($\x =(x_1, ..., x_n) \in \mathbb{R}^{nd}$)
\begin{align}
V_{\bom}^{(n)}(\x)= \sum_{i = 1, ..., n} V_{\bom}^{(1)}(x_i),\qtx{with} V_{\bom}^{(1)}(x)= \sum_{k \in \mathbb{Z}^{d}}  \omega_{k}  \, u(x-k),
\end{align}  
where
\begin{enumerate}
\item 
$\bom=\{ \omega_k \}_{k\in
\Z^d}$ is a family of independent 
identically distributed random
variables  whose  common probability 
distribution $\mu$ has a bounded density $\rho$ and satisfies  $\set{0,M_+}\subset \supp \mu \subseteq 
[0, M_{+}]$ for some $M_{+} > 0$;

\item 
the single site potential $u$ is a measurable function on $\R^d$
 with 
 \begin{equation} \label{u}
u_{-}\Chi_{\Lambda_{\delta_{-}}(0)}\le u \le \Chi_{\Lambda_{\delta_{+}}(0)}\quad \text{for some constants $u_{-}, \delta_{\pm}\in (0,\infty)$}.
\end{equation}
($\Lambda_{\delta_{\pm}}(0)= (-\frac {\delta_{\pm}} 2,\frac {\delta_{\pm}} 2)^d$. We take  $u\le 1$ without loss of generality.)

\end{enumerate}

\item
$U$ is a potential governing the finite range interaction between the $n$ particles.  We take 
\beq\label{2body}
U(\x) = \sum_{1 \leq i < j \leq n } \widetilde{U}(x_i - x_j) ,
\eeq
where $\widetilde{U}\colon \R^d  \to [0, \infty)$ is a bounded measurable function, $\widetilde{U}(y)= \widetilde{U}(-y)$,  with
 $\widetilde{U}(y) = 0$ for $\norm{y}> r_0$ for some $0<r_0 < \infty$.
\end{enumerate}
\end{definition}

\begin{remarks} (i)  The results of this paper are valid if we only assume that the probability measure $\mu$ is uniformly H\"older continuous, i.e., there exist constants $C<\infty$ and $\alpha \in (0,1]$ such  $\mu([a, a+t])\le C t^\alpha$ for all $a\in \R$ and $t\ge 0$.  We assumed that $\mu$ has a bounded density (i.e.,  $\mu$ is uniformly H\"older continuous with $\alpha=1$) for simplicity. 

\noindent{(ii)} We took $U$ to be a two-particle interaction potential  as in \eq{2body} for simplicity.
Our results hold for  nonnegative  bounded finite range   $n$-particle interaction potentials.
\end{remarks}

The $n$-particle Anderson Hamiltonian $H_{\bom}^{(n)}$ is a $\Z^d$-ergodic random Schr\"odinger operator on $L^{2}(\mathbb{R}^{nd})$.  Here   $\Z^d$ acts on $\mathbb{R}^{nd}$  by 
$$(x_1, x_2 \ldots, x_n)\in \mathbb{R}^{nd} \to (x_1 + a, x_2 +a,\ldots, x_n +a)\in \mathbb{R}^{nd}\qtx{for} a \in \Z^d.$$ It follows   (see \cite[Proposition~V.2.4]{CL})
that there exists fixed subsets $\Sigma^{(n)}$,  $\Sigma_{\mathrm{pp}}^{(n)}$, $\Sigma_{\mathrm{ac}}^{(n)}$  and $\Sigma_{\mathrm{sc}}^{(n)}$ of $\R$ so that the spectrum $\sigma(H_{\bom}^{(n)})$
of $H_{\bom}^{(n)}$,  as well as its pure point, 
absolutely continuous, and singular continuous  components,
are equal to these fixed sets with probability one.

Note that $H_{\bom}^{(1)}=H_{0,\bom}^{(1)}$, and it is well known that    $\Sigma^{(1)} =[0,\infty)$  (e.g., \cite{KiM}). It follows, letting $\Sigma_0^{(n)}$ denote the almost sure spectrum of $H_{0,\bom}^{(n)}$, that
\beq\label{Sigma0n}
\Sigma_0^{(n)}= \overline{\Sigma^{(1)}+\ldots + \Sigma^{(1)}}=[0,\infty).
\eeq
In Appendix~\ref{apbottom} we show that  we also have 
\beq\label{spctn}
\Sigma^{(n)}=[0,\infty).
\eeq

\textbf{We now fix a  multi-particle Anderson Hamiltonian $H_{\bom}^{(n)}$, $n\in \N$.}\smallskip

We use the following  definitions and notation:
\begin{enumerate}
\item  Given $\bolda = (a_1, \ldots, a_n) \in \R^{nd}$, we let
  $\scal{\bolda} :=(1 + \norm{\bolda}^2)^{\frac 12}$, $\diam \bolda := \max_{i, \, j = 1, \dots, n} \norm{a_i - a_j}$,  and $\cS_{ \bolda} = \bigl \{a_1, \,...,\, a_n   \bigr\}$.

\item Given $\bolda, \boldb  \in \R^{nd}$, we set $d_{H} ( \bolda, \, \boldb):= d_{H} ( \cS_{ \bolda}, \, \cS_{ \boldb})$, 
where  $d_{H}(S_1, \,S_2)$ denotes the     Hausdorff distance between  finite subsets   $S_1, \, S_2 \subseteq \R^{d}$,  given by
\begin{align}\label{Hausddist}
d_{H}(S_1, \,S_2)& := \max \Bl\{ \max_{x \in S_1} \, \min_{y \in S_2} \norm{x - y} \, , \,   \max_{y \in S_2} \, \min_{x \in S_1} \norm{x - y}  \Br\}\\
&= \max \Bl\{ \max_{x \in S_1} \, \dist(x, \,\,S_2) \, , \,   \max_{y \in S_2} \, \dist(y, \,\,S_1)  \Br\}.
\notag
\end{align}
It follows from the definition that (see \cite{AWmp})
\beq \label{dHdist}
d_{H}( \bolda, \, \boldb )  \leq \norm{\bolda - \boldb} \leq d_{H}( \bolda, \, \boldb )  + \diam\bolda \qtx{for} \bolda, \, \boldb \in \R^{nd} .
\eeq

 \item We fix $\nu_n > \tfrac{nd}{2} $ and let $T_n$ be the operator on $L^{2} \pa{\R^{nd}}$ given by multiplication of the function $\scal{\x} ^{\nu_n}$, where $\scal{\x} = (1 + \norm{\x}^2)^{\frac 12} $.

\item We set $\Chi_{\boldx} = \Chi_{\set{\y \in \R^{nd}; \norm{\y -\x} < \frac 12}}$  for $\boldx \in \R^{nd}$.
  
\end{enumerate}

We prove localization for the  multi-particle  Anderson Hamiltonian $H_{\bom}^{(n)}$, $n\in \N$, as follows.  (Note that $ \Chi_{[0, {E_N})} (H_{\bom}\up{N})=  \Chi_{(-\infty, {E_N})} (H_{\bom}\up{N})$ since $H_{\bom}\up{N}\ge 0$.)

\begin{theorem}\label{localizationH}  Given  $N \in \N$,  there exists an energy  ${E_N}> 0$ such that:

\begin{enumerate}

\item  The following holds with probability one:
\begin{enumerate}
\item \emph{(Anderson Localization)} $H_{\bom}^{(N)}$ has pure point spectrum in the interval $[0, {E_N})$.
Moreover,  there exists $M=M_N>0$ such that
for all $E\in [0, {E_N})$ and  
 $\psi \in  \Chi_{\set{E}} (H_{\bom}^N)$   we have 
\begin{equation}\label{expdecayIntro}
 \|\Chi_{\x} \psi\| \le C_{\bom,E}\norm{T_N^{-1} \psi}\, e^{- M \norm{\x}} \qquad \text{for all}\quad  \x \in \R^{Nd}.
\end{equation}
In particular,  each   {eigenfunction} {$\psi$}   of  $H_{\bom}^{(N)}$
with {eigenvalue}  $E\in [0, {E_N})$ is exponentially localized  with the non-random rate of decay $ M>0$.

\item \emph{(Finite multiplicity of eigenvalues)} The eigenvalues of $H_{\bom}^{(N)}$ in $[0, {E_N})$ have finite multiplicity:  
 \beq
 \tr  \Chi_{\set{E}} (H_{\bom}^{(N)}) <\infty  \qtx{for all} E \in [0, {E_N}).
\eeq

\item \emph{(Summable Uniform Decay of Eigenfunction Correlations (SUDEC))}  For every $\zeta \in \pa{0 , 1}$ there exists a constant       $C_{\bom, \zeta}$  such that for every $E \in [0, {E_N})$ and $\phi, \psi \in \Ran  \Chi_{\set{E}} (H_{\bom}^{(N)})$ we have
\beq
\norm{\Chi_{\x} \phi } \norm{\Chi_{\y} \psi} \leq{ C_{\bom, \zeta}  \norm{T_N^{-1} \phi} \norm{T_N^{-1} \psi} \scal{\x}^{2\nu}   }{  e^{-\pa{d_{H} \pa{\x, \y}}^{\zeta} }  } 
\eeq
for all  $\x, \y \in \R^{Nd}$.
\end{enumerate}

\item \emph{(Dynamical Localization)}  For every $\zeta \in \pa{0, 1}$ and  $\y \in \R^{Nd}$ there exists a constant $C_\zeta(\y)$ such that
\beq
\E \set{\sup_{\abs{g} \leq 1}   \norm{\Chi_{\x}   \Chi_{[0, {E_N})} (H_{\bom}^{(N)}) g( H_{\bom}^{(N)})   \Chi_{\y}  }    } \leq C_\zeta(\y) e^{-\pa{d_{H} \pa{ \x, \y  }}^{\zeta}} 
\eeq
for all  $\x, \y \in \R^{Nd}$,
 the supremum being taken over  Borel functions $g$ on $\R$ with  $ \sup_{t \in \R} \abs{g(t)}\le 1 $.
 In particular, we have
\beq
\E \set{ \sup_{t \in \R}   \norm{\Chi_{\x}  \Chi_{[0, {E_N})} (H_{\bom}^{(N)}) e^{itH_{\bom}^{(N)}} \Chi_{\y}   }  } \leq C_\zeta(\y)e^{-\pa{d_{H} \pa{ \x, \y  }}^{\zeta}}\eeq
for all  $\x \in \R^{Nd}$.

\end{enumerate}

\end{theorem}

\begin{remark*}
SUDEC (Summable Uniform Decay of Eigenfunction Correlations) is equivalent to SULE (Semi Uniformly Localized Eigenfunctions); see  \cite[Remark~3]{GKjsp}.
\end{remark*}

The theorem is proved by a bootstrap multiscale analysis, a statement about  finite volume multi-particle Anderson Hamiltonians. Our finite volumes will be boxes and rectangles, defined as follows:
\begin{enumerate}

\item 
The one-particle box  centered at $x\in\R^d$ with side of length $L >0$ is $\Lambda_L(x)=\set{y \in \R^d; \; \norm{y-x} <  \tfrac{L}{2}}$.  We set $ \widehat{\Lambda} = \Lambda \cap \dspace$.

\item The  $n$-particle box centered at $\boldx\in \R^{nd}$  with side length  $L>0$ is 
   \beq \notag
\nboxLx= \set{\y \in \R^{nd}; \norm{\y-\x} <  \tfrac{L}{2}}= \prod_{i = 1}^{n}   \Lambda_L(x_i);
\eeq
note that $\blam\up{1}_L(x)= \Lambda_L(x)$.  By a box $\boldlambda_L$ in $\R^{nd}$ we mean an $n$-particle box  $\nboxLx$ for some $\boldx \in \R^{nd}$. 
Note that  $\Chi_{\boldx} = \Chi_{\blam_{1}\pa{\boldx}}$  for $\boldx \in \R^{nd}$.

\item  We also define  $n$-particle  rectangles in $\R^{nd}$ centered at points  $\x \in \R^{nd}$:  
  \[ 
 \naboxLx,  \qtx{where}  L_1,L_2,\ldots, L_n>0.
   \]  
(We mostly use $n$-particle boxes, but in a few places we will need $n$-particle  rectangles.)
 
\end{enumerate}

\begin{definition} \label{deffinvol} Given an $n$-particle rectangle $\blam =\naboxLa$, we define the corresponding finite volume Anderson Hamiltonian $H_{\bom, \blam}^{(n)}$  on  $ \L^2(\blam)$ by
\beq 
H_{\bom, \blam}^{(n)} :=H_{0,\bom, \blam}^{(n)} + U_{\blam}, \qtx{with}   H_{0,\bom, \blam}^{(n)}:=-\Delta_{\blam}^{(n)} + V_{\bom, \blam}^{(n)}  ,
\eeq
where $\Delta_{\blam}^{(n)}$ is the Laplacian on $\blam$ with Dirichlet boundary condition, $U_{\blam}$ is the restriction of $U$ to $\blam$, and 
\beq\label{rpn}
V_{\bom , \blam}^{(n)} (\x) = \sum_{i = 1}^{n} V_{\bom, \Lambda_{L_i}(a_i)}^{(1)} (x_i) \qtx{for} \boldx\in \blam ,
\eeq
where $V_{\bom , \Lambda} ^{(1)} $ is defined for a one-particle box $\Lambda \subseteq \R^{d}$ by 
\beq\label{rp1}
V_{\bom , \Lambda} ^{(1)} (x) = \sum_{k \in \widehat{\Lambda}} \omega_{k} \; u(x-k)\qtx{for} x\in \La.
\eeq
We set 
\beq  
R_{\bom, \blam}^{(n)}(z)= (H_{\bom, \blam}^{(n)}  -z)^{-1} \qtx{for} z \notin \sigma \left( H_{\bom, \blam}^{(n)}   \right).
\eeq
 \end{definition}

  Note that  $H_{\bom,\La}^{(1)}=H_{0,\bom,\La}^{(1)}$ and  we have (cf. \eq{Sigma0n}) 
\beq\label{specsumt}
\sigma(H_{0,\bom, \blam}^{(n)})= { \sigma(H^{(1)}_{\bom, \Lambda_{L_1}(a_1)}) + \cdots + \sigma(H^{(1)}_{\bom, \Lambda_{L_n}(a_n)})}.
\eeq
We will often omit the dependency on $n$ from the notation, where it is clear,  and just write  $H_{\bom, \blam} $ for $H_{\bom, \blam}^{(n)} $ and $R_{\bom, \blam}(z)$ for $R_{\bom, \blam}^{(n)}(z)$.

The bootstrap multiscale analysis  uses  three types of good boxes,  defined for a fixed   $\bom$ (omitted from the notation).

\begin{definition} Let $\boldlambda=\nboxLx$ be an $n$-particle box  and let $E\in \R$. 
Let $\theta > 0$, $\zeta \in (0,1)$, and $m>0$.  Then:
\begin{enumerate}
\item 
The  $n$-particle box $\boldlambda$ is  $(\theta, E)$-suitable if, and only if, 
$E \notin \sigma \Bl(H_{\boldlambda}  \Br)$  and
\beq
\norm{\Rlamab} \leq L^{-\theta} \qtx{for all}\bolda, \boldb \in \boldlambda \qtx{with}  \norm{\bolda - \boldb} \geq \tfrac{L}{100} .
\eeq
Otherwise, $\boldlambda$ is called $(\theta, E)$-nonsuitable.

\item The  $n$-particle box $\boldlambda$ is  $(\zeta, E)$-subexponentially suitable (SES) if, and only if, 
$E \notin \sigma \Bl(H_{\boldlambda}  \Br)$  and
\beq
 \norm{\Rlamab}  \leq e^{-L^{\zeta}}
\qtx{for all}\bolda, \boldb \in \boldlambda \qtx{with}  \norm{\bolda - \boldb} \geq \tfrac{L}{100} .
\eeq
Otherwise,    $\boldlambda$   is called   $(\zeta, E)$-nonsubexponentially suitable (nonSES).

\item  The  $n$-particle box $\boldlambda$ is   $(m, E)$-regular if, and only if, 
$E \notin \sigma \Bl(H_{\boldlambda}  \Br)$  and
\beq
 \norm{\Rlamab}  \leq e^{-m\norm{\bolda - \boldb}}
\;\; \text{for all} \;\;\bolda, \boldb \in \boldlambda \;\; \text{with} \; \norm{\bolda - \boldb} \geq \tfrac{L}{100} .
\eeq
Otherwise, $\boldlambda$ is called $(m, E)$-nonregular.

\end{enumerate}
\end{definition}

\begin{remark}  \label{goodbox}
The different types of good boxes are related:
\begin{enumerate}
\item $\nboxLx$  $(m, E)$-regular  \ $\Longrightarrow$  \ $\nboxLx$  
$\left(\tfrac{  mL  }{ 100 \log L  }, E\right)$-suitable.
\item $\nboxLx$  $(\theta, E)$-suitable  \ $\Longrightarrow$  \ $\nboxLx$ \,  $\left(\tfrac{\theta \, \log L}{L}, E\right)$-regular.
\item $\nboxLx$ $\Bl(L^{\zeta-1}, E \Br)$-regular  \ $\Longrightarrow$  \ $\nboxLx$    $\left(\zeta-\tfrac{\log 100}{\log L}, E\right)$-SES.
\item 
$\nboxLx$   $\bigl(\zeta, E \bigr)$-SES  \ $\Longrightarrow$  \ $\nboxLx$   $(L^{\zeta-1}, E)$-regular.
\end{enumerate}
\end{remark}

Our main technical result
 extends the bootstrap multiscale analysis of Germinet and Klein \cite{GK1} (see also \cite{Kl}) to the multi-particle Anderson Hamiltonian.

\begin{theorem}[Bootstrap multiscale analysis]  \label{maintheorem}
 There exist  $p_{0} (n)=p_0(d,n)>0$, $n=1,2,\ldots$, such that, 
 for every   $N \in \N$,  given $\theta > 8Nd$ and  an energy  $E\up{N}> 0$,
there exists ${\cL}=\cL(d,\norm{\rho}_\infty, N, \theta, E\up{N})$,  such that if for some $L_0\ge \cL$ and all  $n=1,2,\ldots,N$ we have 
\beq \label{condpn}
\sup_{\x \in \R ^{nd}} \P  \Bigl\{ \blam_{L_0}^{(n)} (\x) \;\; \text{is} \;\; (\theta,\,E)\text{-nonsuitable} \Bigr\} \leq p_{0}(n)\sqtx{for all} E \le  \En:=2^{N-n}E\up{N},
\eeq
then, given $0< \zeta  <1$, we can find a length scale $L_{\zeta} = L_\zeta(d,\norm{\rho}_\infty, N, \theta,E\up{N},L_0)$,  $\delta_\zeta = \delta_{\zeta}(d,\norm{\rho}_\infty, N, \theta,E\up{N},L_0)>0$, and $m_{\zeta} = m_{\zeta}(\delta_\zeta,L_{\zeta}) > 0$,  so that the following holds for $n = 1, 2, ..., N$:

\begin{enumerate}
\item For every $E \le  \En$,  $L \geq L_{\zeta}$, and $\bolda \in \R^{nd}$,   we have 
\begin{align} 
\P \Bigl\{   \nboxLa  \;\; \text{is} \;\;  \left ({m_\zeta}, \,E \right )\text{-nonregular} \Bigr\}    \leq e^{-L^{\zeta}}.
\end{align}
\item Given  $E_1 < \En$, set  $I(E_1)=[E_1-\delta_{\zeta}, E_1+\delta_{\zeta}]\cap (-\infty, \En]$.  Then, for  every $E_1 < \En$,  $L \geq L_{\zeta}$, and   $\bolda,\boldb \in \R^{nd}$
  with $d_{H} ( \bolda, \, \boldb) \ge L$,  we have  
\begin{align} \label{concmsa}
\P \Bigl\{ \exists \, E \in I(E_1)\;\,\text{so} \;\,  \nboxLa \;\, \text{and} \;\,  \nboxLb  \;\, \text{are} \;\,  \left ({m_\zeta}, \,E \right )\text{-nonregular} \Bigr\} \leq e^{-L^{\zeta}}. 
\end{align}
\end{enumerate}
\end{theorem}

Theorem~\ref{smallcl} shows that the hypotheses of Theorem~\ref{maintheorem} are always satisfied at some nontrivial  interval at the the bottom of the spectrum.

\begin{corollary} [Localization] \label{localization}  Given  $N \in \N$,  an energy  $E\up{N}> 0$, and an  open 
 interval $I\subseteq (-\infty, E\up{N})$,
suppose that the conclusions of Theorem \ref{maintheorem} hold for all energies $E \in I$. Then
 the conclusions of Theorem~\ref{localizationH}  hold on the interval $I$ (i.e., with $I$ substituted for the interval $[0,E_N)$ in Theorem~\ref{localizationH}).
\end{corollary}

 Theorem~\ref{localizationH} follows immediately from Theorem~\ref{maintheorem}, Theorem~\ref{smallcl}, and Corollary~\ref{localization}.

\section{Wegner estimates}

\subsection{Fully and partially separated rectangles}\label{secpfsep}

Let     $\blam=\naboxLa$ be an $n$-particle rectangle.  Given  $\cJ \subseteq \setn $,  we set
\begin{gather}\notag
\boldlambda(\bolda_{\cJ})=\boldlambda^{\cJ}(\bolda_{\cJ}) = \prod_{i \in \cJ} \Lambda_{L_i}(a_i). \qtx{where} \bolda_{\cJ} = (a_i \,\, , \,i \in \cJ),\; \bolda=(\bolda_{\cJ},\bolda_{\cJ^c}); 
\\ \notag
\Pi_{\cJ} \boldlambda^{(n)}(\bolda) = \bigcup_{i \in \cJ} \Lambda_{L_i}(a_i),  \quad \Pi_{i} \boldlambda^{(n)}(\bolda)  = \Pi_{\set{i}} \boldlambda^{(n)}(\bolda) =\Lambda_{L_i}(a_i);\\  \notag
\Pi \boldlambda^{(n)}(\bolda)  = \Pi_{\setn} \boldlambda^{(n)}(\bolda).
\end{gather}

\begin{definition} \label{pfsep}
Let $\nboxx= \prod_{i = 1}^{n}  \,\, \Lambda_{L_i}(x_i)$  and    $\nboxy= \prod_{i = 1}^{n}  \,\, \Lambda_{\ell_i}(y_i)$ be a pair of $n$-particle rectangles.   
\begin{enumerate}
\item  $\nboxx$ and $\nboxy$ are partially separated if, and only if, 
\begin{align} 
\text{either}\quad \Lambda_{L_i}(x_i)\cap  \Pi \nboxy=\emptyset   \qtx{for some}  i \in \setn, \\ \notag
\text{or } 
\ \Lambda_{\ell_j}(y_j)\cap  \Pi \nboxx =\emptyset 
\qtx{for some} j \in \setn.
\end{align}

\item  $\nboxx$ and $\nboxy$ are fully separated if, and only if, 
\beq \label{fullsep}
 \Pi \blam^{(n)}(\boldx) \cap  \Pi \blam^{(n)}(\boldy)   =\emptyset.
\eeq
\end{enumerate}
\end{definition}

Note that, in view of our definition of the finite volume random potentials (see \eq{rpn} and \eq{rp1}), events based on fully separated  rectangles are independent.  Moreover, if the $n$-particle rectangles $\nboxx$ and $\nboxy$ are partially separated, with, say,
$\Lambda_{L_i}(x_i)\cap  \Pi \nboxy=\emptyset $, then events based on $\nboxy$ are independent of the random variables $\set{\omega_k; \ k\in \widehat{\Lambda_{L_i}(x_i)}}$.

\subsection{The Wegner estimates}
Given a one-particle box   ${\Lambda_{L}(x)}$, we will use $\E_{\Lambda_{L}(x)}$ and $\P_{\Lambda_{L}(x)}$  to denote the expectation and probability with respect to the probability distribution of the random variables $\set{\omega_k; \ k\in \widehat{\Lambda_{L}(x)}}$.

\begin{theorem} \label{Wegner0} Let $n \in \N$  and  $E_+ >0$. There exist  constants   
$$ \gamma_{n,E_+}  = \gamma_{n,E_+}(d,M_+, \delta_-, \|\widetilde{U}\|_\infty)>0\sqtx{and}  C_{n,E_+}=C(d,M_+,u_-, \delta_\pm, \|\widetilde{U}\|_\infty, n, E_+),$$  such that, for all  $n$-particle rectangles   $\blam=\naboxLa$  with $\bolda=(a_1,\ldots,a_n)\in \R^{nd}$ and   $  114 \sqrt{nd} \le L_i \le L$ for $i=1,\ldots,n$,  
and all   intervals $I\subseteq [0,E_+)$   with $\abs{I}\le 2\gamma_{n,E_+}$, we have
\beq\label{wegest}
\E_{\Lambda_{L_i}(a_i)}\set{ \tr \Chi_{I} \pa{ H\up{n}_{\bom, \, \boldlambda}   } } \leq C_{n,E_+}\norm{\rho}_{\infty} \abs{I} L^{nd} \qtx{for} i=1,2,\ldots,n.
\eeq
In particular,  for any $ E \leq E_+$,  $0 < {\eps} \le \gamma_{n,E_+}$, and  $i=1,2,\ldots,n$, we have
\beq \label{weggamma}
\P_{\Lambda_{L_i}(a_i)}\Bl\{ \norm{R\up{n}_{\bom,  \boldlambda}(E)} \geq \tfrac{1}{{\eps}} \Br\} = \P_{\Lambda_{L_i}(a_i)} \Bl\{d\,(\sigma(H\up{n}_{\bom, \boldlambda}), E) \leq {\eps} \Br\} \leq 2 C_{n,E_+}\!
\norm{\rho}_\infty\, {\eps} L^{nd}\! .
\eeq
\end{theorem}

We prove Theorem~\ref{Wegner0}  by modifying the proof of \cite[Theorem~1]{HK}.  The main difference  between Theorem~\ref{Wegner0} and  \cite[Theorem~1]{HK} is  that the expectation in \eq{wegest} is taken only with respect to the random variables indexed by the one-particle box $ \Lambda_q $.   This is needed for   proving Corollary~\ref{Wegner2} for a pair of partially  separated $n$-particle rectangles.  Note also that Theorem~\ref{Wegner0}  is proved for arbitrary $n$-particle rectangles, not just $n$-particle boxes  $\nboxLa $ with $\bolda\in \Z^{nd}$ and   $L \in \N$ as in \cite[Theorem~1]{HK}-a consequence of their use of the results of \cite{Kl2}. We extend the results of  \cite{Kl2}  to arbitrary $n$-particle rectangles in Appendix~\ref{apUCP}.

 \begin{proof} Let  $\boldlambda = \prod_{i = 1, \dots, n} \La_i$, where $\La_i= \Lambda_{L_i} (a_i) $, be an  $n$-particle rectangle   with $\bolda\in \R^{nd}$ and   $  114 \sqrt{nd} \le L_i \le L$ for $i=1,\ldots,n$.
Then for $\x\in \La$ we have 
 \begin{align}
 V_{\bom , \blam}^{(n)} (\x) &= \sum_{i = 1}^{n} V_{\bom, \Lambda_{i} }^{(1)} (x_i) =\sum_{i = 1}^{n} \sum_{k \in \widehat{\Lambda_i}} \omega_{k} u(x_i-k)\\ \notag
 & =  \sum_{k \in \Z^d}\omega_k \pa{\sum_{i; \, k\in\widehat{\Lambda_i} } u(x_i -k)}=  \sum_{k \in \Z^d}\omega_k \theta\up{\blam}_k (\boldx),
 \end{align}
 where
 \beq
 \theta\up{\blam}_k (\boldx) = \sum_{\set{i; \, k\in\widehat{\Lambda_i} }} u(x_i -k)\ge  u_{-}  \sum_{\set{i; \, k\in\widehat{\Lambda_i} }}\Chi_{\Lambda\up{1}_{\delta_{-}}(k)} (x_i),
 \eeq
 where we used \eq{u}.  It follows that for $q=1,2,\ldots,n$ we have
 \begin{align} \notag
 H\up{n}_{\bom, \, \boldlambda}&=    -\Delta\up{n}_{\boldlambda} + U_{\boldlambda} + \sum_{k \in \Z^d}\omega_k \theta\up{\blam}_k \\
 & =  -\Delta\up{n}_{\boldlambda} + U_{\boldlambda} + \sum_{k \in \Z^d\setminus \widehat{\La_q}}\omega_k \theta\up{\blam}_k +  \sum_{k \in\widehat{\La_q}}\omega_k \theta\up{\blam}_k. \label{Hdecompq}
 \end{align}

 We now define (with $\eta = \min \set{\frac {\delta_{-}}2, \frac 1 2}$,  $B\up{n}_{ \eta} (\bolds{k})=\set{\x \in \R^{nd}; \ \abs{\x - \bolds{k}}_2 < \eta}$)
 \beq
 W\up{\blam} (\x)= \sum_{\bolds{k}\in \blam \cap \Z^{nd}}\Chi_{B\up{n}_{ \eta} (\bolds{k})} (\x)\qtx{for} \x \in \blam.
 \eeq
 Fix $q \in\set{1, \dots, n}$, and given $\bolds{x}\in  \R^{nd}$, write  $\bolds{x}=(x_q, \bolds{x}^\perp_q)$, where $\bolds{x}^\perp_q \in \R^{(n-1)d}$.  Then
 \beq
 \Chi_{B\up{n}_{ \eta} (\bolds{k})} (\x) \le \Chi_{B\up{1}_{ \eta} (k_q)} (x_q) \Chi_{B\up{n-1}_{ \eta} (\bolds{k}^\perp_q)} (\x^\perp_q) \qtx{for} \bolds{k} \in \Z^{nd}, \;\x \in \R^{nd}.
 \eeq
 We write  $\blam= \La_q \times \blam_q^\perp$, where $\blam_q^\perp= \prod_{i \in \set{ 1, \dots, n}\setminus\set{q}} \La_i $.
 It follows that for all $ \x \in \blam$ we have 
  \begin{align}
W\up{\blam} (\x)&\le \sum_{\bolds{k}\in \blam \cap \Z^{nd}}  \Chi_{B\up{1}_{ \eta} (k_q)} (x_q) \Chi_{B\up{n-1}_{ \eta} (\bolds{k}^\perp_q)} (\x^\perp_q)\\\notag
 &  =  \sum_{k_q \in \widehat{\La_q}}  \Chi_{B\up{1}_{ \eta} (k_q)}(x_q) \set{\sum_{\bolds{k}^\perp_q\in  \blam_q^\perp \cap \Z^{(n-1)d}}   \Chi_{B\up{n-1}_{ \eta} (\bolds{k}^\perp_q)} (\x^\perp_q)}\\ \notag
 & \le   \sum_{k_q \in \widehat{\La_q}}  \Chi_{B\up{1}_{ \eta} (k_q)}(x_q) \le  \sum_{k \in \widehat{\La_q}}  \Chi_{\Lambda\up{1}_{\delta_{-}}(k)}(x_q)\\ \notag
 & \le   \sum_{k \in \widehat{\La_q}}\pa{\sum_{i; \, k\in\widehat{\Lambda_i}}  \Chi_{\Lambda\up{1}_{\delta_{-}}(k)} (x_i)}\le u_-^{-1}  \sum_{k \in \widehat{\La_q}}\theta\up{\blam}_k (\boldx).
 \end{align}

  Fix $E_+>0$. It follows from Theorem~\ref{thmUCPSP} that for any  interval $I\subseteq [0,E_+)$   with $\abs{I}\le 2\gamma_{n,E_+}$ we have  
\begin{align}
 \label{chivchin}
  \Chi_{I}(H\up{n}_{\bom, \boldlambda}) &\le \gamma_{n,E_+}^{-2}\Chi_{I}(H\up{n}_{\bom, \boldlambda}) W\up{\La} \Chi_{I}(H\up{n}_{\bom, \boldlambda}) \\
  \notag &\le   u_-^{-1}  \gamma_{n,E_+}^{-2}\Chi_{I}(H\up{n}_{\bom, \boldlambda}) \pa{ \sum_{k \in \widehat{\La_q}}\theta\up{\blam}_k } \Chi_{I}(H\up{n}_{\bom, \boldlambda}),
 \end{align}
 for  all  $ \bom \in [0,M_+]^{\Z^d}$, where $\gamma_{n,E_+}$ is obtained from \eq{defgamma}: 
\begin{align}\label{defgamman}
\gamma_{n,E_+}^2=\tfrac 1 2   \eta ^{M_{nd} \pa{1 + K^{\frac 2 3}}} \qtx{with} K=  {n(n-1)} \|\widetilde{U}\|_\infty +  2M_+ \delta_+^d +E_+.
\end{align}

The Wegner estimate \eq{wegest} can now be proved following the strategy of \cite[Lemma~3.1]{Kl2},  using \eq{Hdecompq} and \eq{chivchin}. This is what is done in \cite[Proof of Theorem~1]{HK}, the  difference being  that the proof in  \cite{HK} uses a version of \eq{chivchin} where in the right hand side $\pa{ \sum_{k \in \widehat{\La_q}}\theta\up{\blam}_k }$ is replaced by $\pa{\sum_{q=1}^d \sum_{k \in \widehat{\La_q}}\theta\up{\blam}_k }$, and averages over all random variables instead of only  over
the random variables $\set{\omega_i}_{i\in \widehat{\La_q}}$.  The same argument as in \cite{HK}  applies,  using \eq{chivchin} and averaging only over 
the random variables $\set{\omega_i}_{i\in \widehat{\La_q}}$, yielding \eq{wegest}.
 \end{proof}

\begin{corollary}  \label{Wegner2}   Let $n \in \N$  and  $E_+ >0$, and let  $\gamma_{n,E_+}$ be as in  Theorem~\ref {Wegner0}. 
Let 
  $\boldlambda_1 = \prod_{i = 1, \dots, n} \Lambda_{L_i} (a_i) $ and $\boldlambda_2 = \prod_{i = 1, \dots, n} \Lambda_{L^\pr_i} (b_i)  $, with $\bolda,\boldb \in \R^{nd}$ and $  114 \sqrt{nd} \le L_i , L_i^\pr \le L$ for $i=1,\ldots,n$,   be a pair of partially  separated $n$-particle rectangles.    Set
\beq
\widetilde{\sigma} \pa{H_{\blam_1}} = {\sigma} \pa{H_{\blam_1}} \cap (-\infty, E_+]  \; , \hskip 20pt \widetilde{\sigma} \pa{H_{\blam_2}} = {\sigma} \pa{H_{\blam_2}} \cap (-\infty, E_+].
\eeq
 Then there exists a constant $\widetilde{C}_{n,E_+}=\widetilde{C}(d,M_+,u_-, \delta_\pm, \|\widetilde{U}\|_\infty, n, E_+)$,  such that for all  $0<\varepsilon \le \gamma_{n,E_+}$  we have
\beq \label{wegboxes}
\P \Bl\{ \dist \Bl(\widetilde{\sigma} (H_{\boldlambda_{1}}), \widetilde{\sigma} (H_{\boldlambda_{2}}) \Br) \leq {\eps} \Br\} \leq\widetilde{C}_{n,E_+} \norm{\rho}_\infty \,{\eps} L^{2nd} .
\eeq
\end{corollary}

Corollary~~\ref {Wegner2} follows from Theorem~\ref{Wegner0}  in the same way  \cite[Corollary~2.4]{KlN} is derived from \cite[Theorem~2.3]{KlN}. Note that Corollary~\ref {Wegner2} is (up to minor details)  the same as   \cite[Theorem~9]{HK}, although the proofs use somewhat different versions of the Wegner estimate.

\section{Toolkit for the multiscale analysis}\label{sectoolkilt}

\subsection{Deterministic estimates}  The following lemmas are deterministic, i.e., they hold  for a fixed   $\bom$ (omitted from the notation).

 Given $n$-particle boxes $\boldlambda \subseteq \widetilde{\boldlambda}$
we set 
\begin{align} \label{boundary}
& \partial^{\widetilde{\blam}} \blambda := \partial \blam \setminus \partial \widetilde{\blam}, \notag \\
 &\boldlambda ^{\widetilde{\boldlambda}, \; \delta} := \set{ \x \in \blambda \; ; \; \dist (\x, \; \partial^{\widetilde{\blam}} \blambda) \geq \delta}\qtx{for} \delta >0, \\
 &\Upsilon = \Upsilon_{\blam}^{\widetilde{\blam}} := \set{ \x \in \blam \;;\; \dist (\x, \; \partial^{\widetilde{\blam}} \blambda) = \tfrac{1 \; + \; \delta_+}{2}  =: \wdel}, \notag  \\
 &\widetilde{\Upsilon} = \wups_{\blam}^{\wlam}:= \bigcup_{\x \in \Upsilon} \Bigl( \blam_{\tfrac{1}{2}}(\x) \cap \blam \Bigr) .\notag 
\end{align}

\begin{lemma} \label{res-ine}
Let $\blam \subsetneq \wlam$ be two $n$-particles boxes of length $\ell$ and $L$, respectively, with $\ell <  L$, and $z \notin \sigma \pa{H_{\blam}} \cup \sigma \pa{H_{\wlam}}$. Then there exists a constant $C  = C_{n,d}$, such that for $\x \in \blam$ with $\blam_{3+ \delta_+  } (\x) \cap \wlam \subseteq \blam$, and $\y \in \wlam \setminus \blam$, we can find $\bolda \in \ups_{\blam}^{\wlam}$   such that
\begin{align} \label{resolventin}
\norm{\Chiy R_{\wlam}(z) \Chix} &\leq C_{n,d} \,\ell^{nd-1} \sqrt{ 5 +  \max \{ 0, \Re (z)  \}}    \norm{ \Chiy  R_{\wlam} (z) \Chi_{\bolda}  }  
\norm{\Chi_{\bolda} R_{\blam} (z)  \Chix }. 
\end{align}
In particular, if   $\Re(z) \leq E\up{n}$, for some fixed energy $E\up{n}$, we get 
\beq \label{resolventine}
\norm{\Chiy R_{\wlam}(z) \Chix} \leq \ell^{nd} \norm{ \Chiy  R_{\wlam} (z) \Chi_{\bolda}  }  
\norm{\Chi_{\bolda} R_{\blam} (z)  \Chix },
\eeq
provided $\ell$ is sufficiently large (depending on $ E\up{n}$).
\end{lemma}

 Lemma \ref{res-ine} is just   \cite[Lemma~2.4(i)]{GKber} with minor modifications.

\begin{lemma} \label{NE}
Given an $n$-particle box, $\blam$ of side $\ell$, for every $E \geq 0$ we have
\begin{align}
\#\set{\lambda \in \sigma \pa{H_{\blam}^{(n)}} \cap (-\infty, E]   } &= \tr \set{\Chi_{(-\infty, \, E]} \pa{H_{\blam}\up{n}}} \leq C_{nd}  E^{\tfrac{nd}{2}} \ell^{nd} .
\end{align}
\end{lemma}

 Lemma \ref{NE} follows from   \cite[Lemma~3.3]{KK} (see also \cite[Eq.~(A.7)]{GK5}). 
 
We also use the following Combes-Thomas estimate from \cite[Eq. (19) in Theorem 1]{GK6}.
\begin{lemma}  \label{cthomas}
Let $\blam$ be an $n$-particle box. Then for all  $E < \inf \sigma \pa{H_{\blam}}$
we have
\begin{align} 
\norm{\Rlamxy} & \leq \tfrac 4 3 \pa{\inf \sigma \pa{H_{\blam}}-E}^{-1} e^{- \frac 1 2 \sqrt{{\inf \sigma \pa{H_{\blam}}-E}}\, \pa{\norm{\x - \y}-nd} } 
\end{align}
for all $\x,\y \in \R^{nd}$. In particular, if we take $\norm{\x - \y} \geq \tfrac{L}{100}$ with $L$ sufficiently large, we have 
\begin{align} \label{ctho}
\norm{\Rlamxy} & \leq \tfrac 4 3 \pa{\inf \sigma \pa{H_{\blam}}-E}^{-1} e^{- \frac{1}{ 3} \sqrt{{\inf \sigma \pa{H_{\blam}}-E}}\, \norm{\x - \y}} \\
&\leq \tfrac 4 3 \pa{\inf \sigma \pa{H_{\blam}}-E}^{-1} e^{- \frac{\sqrt{{\inf \sigma \pa{H_{\blam}}-E}}}{300}L}. \notag
\end{align}
\end{lemma}
\subsection{Suitable cover}

Following  \cite[Definition~3.12]{GKber} we introduce   suitable covers of $n$-particle boxes.

\begin{definition}\label{defsuitcov}
Given scales $\ell <  L$,  a  \emph{suitable $\ell$-covering} of a box $\NboxLx$ is
a collection of boxes $\blam_{\ell}$ of the form
\begin{align}\label{standardcover}
\cG_{\NboxLx}^{(\ell)}= \{ \blam_{\ell}(\mathbf{r})\}_{\mathbf{r} \in \G_{\NboxLx}^{(\ell)}},
\end{align}
where
\beq  \label{bbG}
\G_{\NboxLx}^{(\ell)}:= \{ \x + \alpha\ell  \Z^{Nd}\}\cap \NboxLx \quad 
\text{with}  \quad \alpha \in \br{\tfrac {3} {5},\tfrac {4} {5}}   \cap \set{\tfrac {L-\ell}{2 \ell n}; \, n \in \N }.
\eeq
\end{definition}

Suitable covers are useful because of \cite[Lemma~3.13]{GKber}, stated below. 

\begin{lemma}\label{lemcovering} 
Let $\ell \le \frac  L   6$. Then every box   $\NboxLx$  has a   suitable $\ell$-covering, and   for any  suitable $\ell$-covering
  $\cG_{\NboxLx}^{(\ell)}$  of $\NboxLx$ we have 
  \begin{align}\label{nestingproperty}
&\NboxLx =\bigcup_{\boldr \in\G_{\NboxLx}^{(\ell)} } \blam_{\ell}(\boldr),\\ 
\label{bdrycover}
&\text{for each $\y \in \NboxLx$ there is $\boldr \in \G_{\NboxLx}^{(\ell)}$ with $\blam_{\frac{\ell} 5}(y)\cap \NboxLx \subset \blam_{\ell}(\boldr)$},\\
\label{freeguarantee}
&\blam_{\frac{\ell}{5}}(\boldr)\cap \blam_{\ell}(\boldr^{\prime})=\emptyset
\quad \text{for all} \quad \boldr,\boldr^{\pr}\in  \x + \alpha\ell  \Z^{d}, \  \boldr\ne \boldr^{\prime},\\ \label{number}
&  \pa{\tfrac{L} {\ell}}^{Nd}\le \#\G_{\NboxLx}^{(\ell)}= \pa{ \tfrac{L-\ell} {\alpha \ell}+1}^{Nd} \le   \pa{\tfrac{2L} {\ell}}^{Nd}.
\end{align}
Moreover,  given $\y \in  \x + \alpha \ell  \Z^{Nd}$ and $k \in \N$, it follows that
\beq \label{nesting}
\blam_{(2  k \alpha  +1) \ell}(\y)= \bigcup_{\boldr \in  \{ \x + \alpha \ell \Z^{Nd}\}\cap \blam_{(2k \alpha  +1) \ell}(\y) } \blam_{\ell}(\boldr),
\eeq
and  $ \{ \blam_{\ell}(\boldr)\}_{\boldr \in  \{ \x + \alpha\ell  \Z^{Nd}\}\cap\blam_{(2k\alpha +1) \ell}(\y)}$ is a suitable $\ell$-covering of the box $\blam_{(2 k\alpha +1) \ell}(\y)$.  In particular,
\beq \label{ell5cover}
\text{for each $\y \in \Z^{Nd}$ there is $\boldr \in  \x + \alpha\ell  \Z^{Nd}$ with $\blam_{\frac{\ell} 5}(\y) \subset \blam_{\ell}(\boldr)$}.
\eeq

\end{lemma}

\begin{remark} \label{remcovering}
In performing the $N$-particle multiscale analysis, we will
utilize Lemma \ref{lemcovering}  in the following way: we first choose some appropriate $k_1$ such that $2k_1\alpha+1 > 3N$. Given $J \in \N$, let $\bolda_1, \cdots, \bolda_t \in \NboxLx$ where $t \leq J N^N$, 
there exists $\y_1, \cdots, \y_t \in \x + \alpha\ell  \Z^{Nd}$ such that $\blam_{3N\ell}(\bolda_1) \cap \NboxLx \subseteq \blam_{\pa{2k_1\alpha+1}\ell}(\y_1) \subseteq \NboxLx, \cdots, \blam_{3N\ell}(\bolda_t) \cap \NboxLx \subseteq \blam_{\pa{2k_1\alpha+1}\ell}(\y_t) \subseteq \NboxLx$.

If the set $\blam_{(2  k_1 \alpha  +1) \ell}(\y^{(1)}) \cup \cdots \cup \blam_{(2  k_1 \alpha  +1) \ell}(\y^{(m)}) \subseteq \NboxLx$ is connected, where $\y^{(1)}, \cdots, \y^{(m)} \in \x + \alpha\ell  \Z^{Nd}$,   then we will take $k_m$ to be the smallest integer such that $\blam_{(2  k_1 \alpha  +1) \ell}(\y^{(1)}) \cup \cdots \cup \blam_{(2  k_1 \alpha  +1) \ell}(\y^{(m)}) \subseteq  \blam_{(2  k_m \alpha  +1) \ell}(\boldr) \subseteq \NboxLx$, for some $\boldr \in  \x + \alpha\ell  \Z^{Nd}$. 
Moreover, for each $\y \in \Z^{Nd}$, we set $\blam_{\ell}^{(\boldy)} = \blam_{\ell}(\boldr)$ where $\blam_{\ell}(\boldr)$ comes from \eq{ell5cover}. 
\end{remark}

\subsection{Partially  and fully interactive boxes}\label{secprint}

Following Chulaevsky and Suhov \cite{CS2,CS3}, we divide  $n$-particle boxes into two types: partially interactive and fully interactive.

\begin{definition} \label{PIFIbox}
An $n$-particle box $\nboxLa$ is said to be partially interactive (PI) if and only if there exists $\emptyset\ne \cJ \subsetneq \setn$ such that
 $\nboxLa \subseteq \cE_{\cJ}$, where $\cE_{\cJ} = \Bl\{ \boldx \in \R^{nd} \,\, \vert \,\, \min_{i \in \cJ, \, j \notin \cJ} \norm{x_i - x_j} > r_0   \Br\}$.\
 
If $\nboxLa$ is not partially interactive, it is said to be  fully interactive (FI).
\end{definition}

If the $n$-particle box $\nboxLa$ is partially interactive, by writing
 $\nboxLa = \boldlambda_{L}^{{\cJ}}(\bold{a}_{\cJ}) \times  \boldlambda_{L}^{{\cJ}^{c}}(\bold{a}_{\cJ^{c}} ) $  we are implicitly stating that   $\nboxLa \subseteq \cE_{\cJ}$ for  $\emptyset\ne \cJ \subsetneq \setn$.  
We  set  $\sigma_\cJ =\sigma \left( H_{\boldlambda_{L}^{{\cJ}}(\bold{a}_{\cJ})}  \right)  $ and $\sigma_{\cJ^{c}}=\sigma \left( H_{ \boldlambda_{L}^{{\cJ}^{c}}(\bold{a}_{\cJ^{c}} )  }  \right)$.
Given $\lambda \in \sigma_\cJ$, we write $P^{\cJ}_\lambda= \Chi_{\set{\lambda}}\pa{H_{\boldlambda_{L}^{{\cJ}}(\bold{a}_{\cJ})} }$.

\begin{lemma} \label{PIpro}
Let $\nboxLu = \boldlambda_{L}^{{\cJ}}(\bold{u}_{\cJ}) \times  \boldlambda_{L}^{{\cJ}^{c}}(\bold{u}_{\cJ^{c}} ) $ be a PI $n$-particle box. Then:

\begin{enumerate}
\item $ \Pi_{\cJ} \nboxLu  \,\, \bigcap \,\, \Pi_{\cJ^{c}} \nboxLu = \emptyset$, so events based on $ \boldlambda_{L}^{{\cJ}}(\bold{u}_{\cJ})$ and  $ \boldlambda_{L}^{{\cJ}^{c}}(\bold{u}_{\cJ^{c}} )  $ are independent.  
\item $H_{\nboxLu} = H_{\bold{\Lambda}_{L}^{{\cJ}}(\bold{u}_{\cJ})} \otimes I_{\bold{\Lambda}_{L}^{ {\cJ}^{c}}(\bold{u}_{\cJ^{c}})} + I_{\bold{\Lambda}_{L}^{{\cJ}}(\bold{u}_{\cJ})} \otimes H_{\bold{\Lambda}_{L}^{ {\cJ}^{c}}(\bold{u}_{\cJ^{c}})}$.

\item  $\sigma \left( H_{\nboxLu} \right) = \sigma \left( H_{\boldlambda_{L}(\bold{u_{\cJ}})}  \right) + \sigma \left( H_{ \boldlambda_{L}(\bold{u_{\cJ^{c}}})}  \right)$. 
 
\item  If $z \notin \sigma \left( H_{\nboxLu} \right)$, we have
\begin{align}
 R_{\nboxLu} (z)  &= \sum_{\lambda \in \sigma _{\cJ}}  \sum_{\mu \in \sigma _{\cJ^{c}}} \tfrac{1}{\lambda + \mu - z} P^{\cJ}_{\lambda} \otimes P^{\cJ^c}_{\mu}\\
& =\sum_{\lambda \in \sigma _{\cJ}}  P^{\cJ}_{\lambda} \otimes 
  R_{ \boldlambda_{L}(\bold{u_{\cJ^{c}}})} (z- \lambda )  = \sum_{\mu \in \sigma _{\cJ^{c}}} 
  R_{  \boldlambda_{L}(\bold{u_{\cJ}})  } (z- \mu )   \otimes P^{\cJ^c}_{\mu},\notag
\end{align}
and for all $\x,\y \in 
\nboxLu$ we get
\begin{align}\label{Rineqz}
\norm{\Rlamxyz} &\leq \sum_{\lambda \in \sigma_{\cJ}} \norm{\Chi_{\x_{\cJ^{c}}}  R_{ \boldlambda_{L}(\bold{u_{\cJ^{c}}})} (z- \lambda )  \Chi_{\y_{\cJ^{c}}}    },\\
\norm{\Rlamxyz} &\leq \sum_{\mu \in \sigma_{\cJ^{c}}} \norm{\Chi_{\x_{\cJ}}  R_{  \boldlambda_{L}(\bold{u_{\cJ}})  } (z- \mu )  \Chi_{\y_{\cJ}}    }.\notag
\end{align}
\end{enumerate}
 \end{lemma}

\begin{lemma}\label{PIsuit}
 Let $\nboxlu = \boldlambda_{\ell}^{{\cJ}}(\bold{u}_{\cJ}) \times  \boldlambda_{\ell}^{{\cJ}^{c}}(\bold{u}_{\cJ^{c}} ) $ be a PI $n$-particle box and  $E \le E^{(n)}$. If $\ell$ is sufficiently large, the following holds:
\begin{enumerate}

\item   Given $\theta>2nd+2$, suppose  $\boldlambda_{\ell}^{{\cJ}}(\bold{u_{\cJ}})$ is $(\theta, \,E - \mu  )$-suitable  for every $\mu \in \sigma_{\cJ^c} \cap (-\infty, \,2E^{(n)}]$ and  $\boldlambda_{\ell}^{{\cJ^c}}(\bold{u_{\cJ^{c}}})$ is $(\theta, \,E - \lambda  )$-suitable  for every $\lambda \in \sigma_{\cJ} \cap (-\infty, \,2E^{(n)}]$.  Then  $\nboxlu$ is $\left(\tfrac{\theta}{2}, \, E \right)$-suitable.

\item   Given $0 < m \le \frac{1}{6}  \sqrt{E^{(n) }}$, 
suppose  $\boldlambda_{\ell}^{{\cJ}}(\bold{u_{\cJ}})$ is $(m,\, E - \mu)$-regular  for every $\mu \in \sigma_{\cJ^c} \cap (-\infty, \,2\,E^{(n)}]$ and  $\boldlambda_{\ell}^{{\cJ^c}}(\bold{u_{\cJ^{c}}})$ is $(m,\, E - \lambda)$-regular  for every $\lambda \in \sigma_{\cJ} \cap (-\infty, \,2\,E^{(n)}]$.  Then  $\nboxlu$ is $\pa{ m - \tfrac{100 \pa{nd+1} \log \pa{2\ell}}{\ell}, \, E }$-regular.

\item   Given $0 < \zeta^{\prime} < \zeta < 1$, suppose  $\boldlambda_{\ell}^{{\cJ}}(\bold{u_{\cJ}})$ is $(\zeta,\, E - \mu )$-SES  for every $\mu \in \sigma_{\cJ^c} \cap (-\infty, \,2\,E^{(n)}]$ and  $\boldlambda_{\ell}^{{\cJ^c}}(\bold{u_{\cJ^{c}}})$ is $(\zeta,\, E - \lambda)$-SES  for every $\lambda \in \sigma_{\cJ} \cap (-\infty, \,2\,E^{(n)}]$.  Then  $\nboxlu$ is $\left(\zeta^{\prime}, \, E \right)$-SES.

\end{enumerate}
\end{lemma}

\begin{proof}  
We prove (ii),   the proofs of (i) and  (iii) are similar. 
Given $\boldx, \boldy \in \Nboxlu$ with $\norm{\boldx - \boldy} \geq \tfrac{\ell}{100},$ then either we have $\norm{\boldx_{\cJ} - \boldy_{\cJ}} \geq \tfrac{\ell}{100},$ or $\norm{\boldx_{\cJ^{c}} - \boldy_{\cJ^{c}}} \geq \tfrac{\ell}{100}.$ Without loss of generality, we suppose that $\norm{\boldx_{\cJ} - \boldy_{\cJ}} \geq \tfrac{\ell}{100}$.  Then, using \eq{Rineqz}, Lemma \ref{NE}, and  that $\blam_{\cJ}= \boldlambda_{\ell}^{{\cJ}}(\bold{u_{\cJ}})$ is $(m,\, E - \mu)$-regular  for every $\mu \in \sigma_{\cJ^c} \cap (-\infty, \,2E^{(n)}]$, setting $\sigma_{\cJ^c}(E) = \sigma_{\cJ^c} \cap (-\infty, \,E]$, we get

\begin{align} \label{begest}
&\norm{\Rlamxy} \leq \sum_{\mu \in \sigma_{\cJ^{c}}} \norm{\Chi_{\x_{\cJ}}  R_{\blam_{\cJ}}  (E - \mu)  \Chi_{\y_{\cJ}}    } \\
& = \!\!\sum_{ \mu \in \sigma_{\cJ^c}(2\En) } \norm{\Chi_{\x_{\cJ}}  R_{\blam_{\cJ}}  (E - \mu)  \Chi_{\y_{\cJ}}    } + \sum_{\mu \in \sigma_{\cJ^{c}}\setminus\sigma_{\cJ^c}(2\En)} \!\!\norm{\Chi_{\x_{\cJ}}  R_{\blam_{\cJ}}  (E - \mu)  \Chi_{\y_{\cJ}}    } \notag \\
&\leq {C_{nd} \pa{ 2E^{(n)}}^{\tfrac{nd}{2}}\ell^{nd}}{e^{-m \norm{\x_{\cJ} - \y_{\cJ}}}} + \sum_{\mu \in\sigma_{\cJ^{c}}\setminus\sigma_{\cJ^c}(2\En)}\!\! \norm{\Chi_{\x_{\cJ}}  R_{\blam_{\cJ}}  (E - \mu)  \Chi_{\y_{\cJ}}    }  \notag\\
&\leq {\ell^{nd + 1}}{e^{-m \norm{\x_{\cJ} - \y_{\cJ}}}}  + \sum_{k = 2}^{\infty}  \;\;\sum_{\substack{\mu \in \sigma_{\cJ^{c}} \\ k\,E^{(n)} < \mu \leq  \pa{k+1}\,E^{(n)}    }} \norm{\Chi_{\x_{\cJ}}  R_{\blam_{\cJ}}  (E - \mu)  \Chi_{\y_{\cJ}}    }. \notag 
\end{align}

Applying \eqref{ctho} for $\mu \in \sigma_{\cJ^{c}}$ with $k\,E^{(n)} <  \mu \;\leq  \pa{k+1} \,E^{(n)}$ we get
\begin{align}
\norm{\Chi_{\x_{\cJ}}  R_{\blam_{\cJ}}  (E - \mu)  \Chi_{\y_{\cJ}}    }& \leq \tfrac{4}{3} \pa{ \inf \sigma_{\cJ}  - \pa{E-\mu}  }^{-1} e^{-\frac{1}{3}  \sqrt{\inf \sigma_{\cJ}   - \pa{E-\mu}  } \norm{\x_{\cJ} - \y_{\cJ}}} \notag \\ \label{kRJ}
& \leq \tfrac{4}{3} \pa{ kE^{(n)} - E  }^{-1} e^{-\frac{1}{3}  \sqrt{kE^{(n)} - E } \norm{\x_{\cJ} - \y_{\cJ}}}\\
& \leq \tfrac{4}{3} \pa{ (k-1)E^{(n)}  }^{-1} e^{-\frac{1}{3}  \sqrt{(k-1)E^{(n) } }\norm{\x_{\cJ} - \y_{\cJ}}}.  \notag
\end{align}
Using  Lemma \ref{NE} and \eq{kRJ}, we have ($k\ge 2$, $\norm{\boldx_{\cJ} - \boldy_{\cJ}} \geq \tfrac{\ell}{100}$)
\begin{align}
&\sum_{\substack{\mu \in \sigma_{\cJ^{c}} \\ k E^{(n)} < \mu \leq  \pa{k+1}E^{(n)}    }}\norm{\Chi_{\x_{\cJ}}  R_{\blam_{\cJ}}  (E - \mu)  \Chi_{\y_{\cJ}}    } \\ \notag
&  \hskip80pt \leq C_{nd}\ell^{nd} \pa{\pa{k+1}\,E^{(n)} }^{\tfrac{nd}{2}-1}\e^{-\frac{1}{3}  \sqrt{(k-1)E^{(n) } }\norm{\x_{\cJ} - \y_{\cJ}}}\\
\notag  & \hskip80pt \le e^{-\frac{1}{6}  \sqrt{(k-1)E^{(n) } }\norm{\x_{\cJ} - \y_{\cJ}}},
\end{align}
for sufficiently large $\ell$, so 
\begin{align}
&\sum_{k = 2}^{\infty}  \;\;\sum_{\substack{\mu \in \sigma_{\cJ^{c}} \\ k\,E^{(n)} < \mu \leq  \pa{k+1}\,E^{(n)}    }} \norm{\Chi_{\x_{\cJ}}  R_{\blam_{\cJ}}  (E - \mu)  \Chi_{\y_{\cJ}}    }\\ \notag
& \qquad  \qquad \le    \sum_{k = 2}^{\infty} e^{-\frac{1}{6}  \sqrt{(k-1)E^{(n) } }\norm{\x_{\cJ} - \y_{\cJ}}}\le 
2 e^{-\frac{1}{6}  \sqrt{E^{(n) } }\norm{\x_{\cJ} - \y_{\cJ}}},
\end{align}
for $\ell$ large.  Using \eq{begest} and $m \le \frac{1}{6}  \sqrt{E^{(n) }}$, we get
\begin{align}
\norm{\Rlamxy}& \leq  {\ell^{nd + 1}}{e^{-m \norm{\x_{\cJ} - \y_{\cJ}}}}  + 2 e^{-\frac{1}{6}  \sqrt{E^{(n) } }\norm{\x_{\cJ} - \y_{\cJ}}}& \\
& \le 2 {\ell^{nd + 1}}{e^{-m \norm{\x_{\cJ} - \y_{\cJ}}}}  \leq e^{- \pa{ m - \tfrac{100 (nd+1) \log (2\ell) }{\ell}   } \norm{\x - \y}   }.   \notag
\end{align}
\end{proof}

\begin{definition}  \label{L distant}
Let $\nboxLa$ and $\nboxLb$ be a pair of  $n$-particle boxes.  We say  
$\nboxLa$ and $\nboxLb$ are $L$-distant if  and only if 
\beq
\max \bigl\{\dist \, ( \,  \boldb, \, \cS_{\bolda}^{n} ) , \,\, \dist \, ( \, \bolda, \, \cS_{\boldb}^{n})   \bigr\} \ge  3nL .
\eeq
\end{definition}

The following lemma gives a sufficient condition for a pair of FI $n$-particle boxes to be fully separated, and hence for events based on these boxes to be independent. We omit the proof.

\begin{lemma} \label{part1prop2} Let   
 $\nboxLa$ and $\nboxLb$ be a pair of FI $n$-particle boxes, where $L$ is sufficiently large.  Then
 $\nboxLa$ and $\nboxLb$ are fully separated if
\beq\label{FIsep}
\max_{x \in \cS_{ \bolda},y \in  \cS_{\boldb}} \norm{x-y} \ge 3nL .
\eeq
In particular,  a pair of  $L$-distant FI $n$-particle boxes are fully separated.
 \end{lemma}

\subsection{Resonant rectangles}
\begin{definition} Let $\boldlambda = \prod_{i = 1, \dots, n} \Lambda_{L_i} (a_i)$ with $L = \min_{i =1,.., n}\set{L_i}>0$ be an $n$-particle rectangle,  $E\in \R$, $s > 0$, and $\beta \in (0, \, 1)$. 
\begin{enumerate}
\item  $\boldlambda$ is called  $(E,s)$-suitably resonant provided $\dist \pa{\sigma \pa{H_{\boldlambda}^{(n)}} , E } < L^{-s}$.  Otherwise, $\boldlambda$ is said to be $(E,s)$-suitably nonresonant.

\item $\boldlambda$ is called $(E,\beta)$-resonant provided $\dist \Bl( \sigma \bigl( H_{\boldlambda}^{(n)} \bigr) , E \Br) < \tfrac{1}{2}  e^{-L^{\beta}} $. Otherwise, $\boldlambda$ is said to be $(E,\beta)$-nonresonant.
\end{enumerate}
\end{definition}

\section{The  initial  step for the bootstrap multiscale analysis}\label{secinitialMSA}
 
We now show that the hypotheses of Theorem~\ref{maintheorem} are verified for energies  at the bottom of the spectrum.  Recall $\Sigma\up{n}=[0,\infty)$.

\begin{theorem}\label{smallcl} Let  $\theta>0$ and $0<p_0<1 $, and fix  $\eps >0$. Then for all $n\in \N$ there exists  $\cL_n=\cL_n (d, u_-, \delta_\pm,\mu,\theta,p_0,\eps)$, such that for all $L \ge \cL_n$ and  $\x\in \R^{nd}$ we have
\beq\label{initialsuit2}
\P\set{\nboxLx \sqtx{is} (\theta,\,E)\text{-suitable}}\ge 1 -  p_0 \qtx{for all} E \le E_L\up{n},
\eeq
where
\beq\label{EL}
E_L\up{n}= \tfrac n 2 \pa{d\log (L+\delta_+ +2) -\log p_0 +\log n}^{-\frac {2+\eps}d}.
\eeq

\end{theorem}

\begin{proof}

  We start with a well known result for the one-particle case. Fix $\theta >0$, $p_0>0$, and $n \in \N$,  $\eps >0$, and set $p_n =\frac {p_0} n$.  As shown in \cite[Proof of Proposition~4.3]{GKber}, there exists an energy $E_1=E_1(d, u_-, \delta_-,\mu,\eps) >0$ such that for energies $E\le E_1$, $x\in \R^d$, and scales $L \in 2\N$, we have
  \beq
  \P\set{\sigma\pa{H_{\bom, \boxLx}^{(1)}}\cap (-\infty,E]\ne \emptyset}\le \e^{-E^{-\frac d{2+\eps}}}L^d,
  \eeq
 and hence
 \beq
 \P\set{H_{\bom, \boxLx}^{(1)}\ge \min\set{ \pa{d\log L -\log p_n}^{-\frac {2+\eps}d},E_1}} \ge 1 -p_n.
 \eeq

  Proceeding as in \cite[Proof of Proposition~4.3]{GKber}, for
each  $x \in \R^d$ and scales $L\ge 1$, we consider  the event
\beq
\Omega_{L,x}= \set{H_{\bom, \boxLx}^{(1)} \geq  2 E^\pr_L}, \qtx{where} E^\pr_L=\tfrac 12  \pa{d\log (L+\delta_+ +2) -\log p_n}^{-\frac {2+\eps}d},
\eeq
and conclude that 
 for scales $L \ge \cL^\pr_n= \cL^\pr_n (d, u_-, \delta_-,\mu,p_0,\eps)$ we have
 \beq
\P\set{\Omega_{L,x}}\ge 1 -  p_n \qtx{for all} x \in \R^d.
\eeq

Now let    $\x\in \R^{nd}$, and consider the   $n$-particle box $\nboxLx$.   Given  $L\ge 1$,  we set
\beq
\Omega_{L,\x}= \bigcap_{i=1}^n \Omega_{L,x_i}, \qtx{so}  \P\set{\Omega_{L,\x}}\ge 1 -  n p_n=1-p_0\qtx{for}  L \ge \cL_n^\pr.
\eeq
In view of \eq{specsumt}
 we have
\beq
\inf \sigma \pa{H_{0,\nboxLx}  } \geq 2nE^\pr_L =2 E_L\up{n}  \qtx{for all } \bom \in \Omega_{L,\x},
\eeq 
which implies, using $U \geq 0$, that
\beq\label{nlowerb}
\inf \sigma \pa{ H_{\nboxLx}  } \geq 2E_L\up{n}  \qtx{for all} \bom \in\Omega_{L,\x}.
\eeq

We now fix  $\bom \in\Omega_{L,\x}$ and let  $E \le E_L\up{n}$ and   $\bolda, \boldb \in  \blam=\nboxLx$ with $\norm{\bolda - \boldb} \geq \tfrac{L}{100}$. It follows from \eq{nlowerb} and Lemma \ref{cthomas} that $E \notin \sigma \pa{ H_{\nboxLx}  }$ and (provided $L$ is sufficiently large)
\beq \label{beghyp}
\norm{\Rlamab} \leq \tfrac 4 3 \pa{E_L\up{n}}^{-1} e^{- \frac{L}{ 201} \sqrt{E_L\up{n}}}. 
\eeq

Thus, given $\theta>0$, there exists $\cL_n=\cL_n (d, u_-, \delta_\pm,\mu,\theta,p_0,\eps)$, such that for all $L \ge \cL_n$ and  $\x\in \R^{nd}$ we have \eq{initialsuit2}.
 \end{proof}

\begin{remark} 
The hypotheses of Theorem~\ref{maintheorem} can  be verified in a fixed interval at the bottom of the spectrum at high disorder.
To see that, consider $H^{(n)}_{\bom,\lambda}=  -\Delta^{(n)} + \lambda V_{\bom}^{(n)} + U $, where $V_{\bom}^{(n)}$ and $ U$ are as in Definition~\ref{defAndmodel} and $\lambda >0$ is the disorder parameter. $H^{(n)}_{\bom,\lambda}$ can be rewritten as $n$-particle Anderson Hamiltonian   in the form of Definition~\ref{defAndmodel} by replacing the probability distribution $\mu$ by the probability distribution $\mu_{\lambda}$, where
$\mu_{\lambda}$ is the probability distribution of the random variable $\lambda \omega_0$, that is, 
$
\mu_{\lambda}(B) = \mu\pa{\lambda^{-1}B}$  for all Borels sets $ B\subset \R$.   In particular, $\mu_{\lambda}$ has density $\rho_\lambda (\omega_0)=  \lambda^{-1}\rho\pa{\lambda^{-1}\omega_0 }$.

For simplicity we   assume 
the covering condition 
\beq \label{covercrit}
U_{-} \Chi_{\Lambda} \leq \sum_{k \in \dspace \cap \Lambda} u(x-k)            
\eeq 
for all one-particle boxes $\Lambda$, where   $U_{-} > 0$.  (The condition \eq{covercrit}    can be guaranteed by requiring   $  \delta_{-}\ge 2$.   If we restrict ourselves to boxes $\boxLx$ with $x\in \Z^d$ and $L$ an odd natural number it suffices to require $\delta_- \ge 1$.)  
In this case it is  well known how to proceed in the one-particle case
(see \cite{CH,GK4}):  Given $E_1>0$, it follows from \eq{covercrit} that
\beq
\P\set{H_{\bom,\lambda, \boxLx}^{(1)} \geq 2E_1} \ge 1 - L^d\mu_\lambda \set{[0,2E_1U_-^{-1}]}\ge 1- 2 E_1 U_-^{-1} \lambda^{-1}\norm{\rho}_{\infty} L^{d}.
\eeq
Proceeding as in the proof of Theorem~\ref{smallcl}, we obtain
\beq
\P\set{H_{\bom, \lambda,\nboxLx}^{(n)} \geq 2nE_1 } \ge  1- 2n E_1U_-^{-1} \lambda^{-1}\norm{\rho}_{\infty} L^{d}.
\eeq
Given $0<p(n)<1$ and $E>0$, we set
\beq
  \lambda(E,L,p(n)) = \frac {2 E  \norm{\rho}_{\infty}L^{d}} {p(n)U_-} ,
\eeq
obtaining for all $E>0$
\beq
\P\set{H_{\bom, \lambda,\nboxLx}^{(n)} \geq 2E} \ge  1- p(n) \qtx{for all} \lambda \ge  \lambda(E ,L,p(n)) .
\eeq
To use Lemma \ref{cthomas} as in \eq{beghyp}, we require
\beq
\tfrac 4 3 {E}^{-1} e^{- \frac{L}{ 201} \sqrt{E}} \le L^{-\theta}, \qtx{i.e.,} L \ge L(E,\theta).
\eeq
We conclude that, given $0<p(n)$, $E_0 >0$,  and $\theta >0$, for all   $L \ge L(E_0,\theta)$ and  $\lambda \ge  \lambda(E_0 ,L,p(n))$ we have 
\beq\label{initialsuit2bb}
\P\set{\nboxLx \sqtx{is} (\theta,\,E)\text{-suitable for}\;  H^{(n)}_{\bom, \lambda}}\ge 1 - p(n) \qtx{for all} E \le E_0.
\eeq

If we do not assume the covering condition \eq{covercrit}, we can still prove a large disorder result using \cite[Proposition~4.5]{GKber} for the one-particle case.

\end{remark}

\section{The multi-particle bootstrap multiscale analysis}\label{secMSAproof}

 Theorem~\ref{maintheorem}  is proven  by  induction on $N$,  the number of particles.  For $N=1$ the theorem      was proved   by Germinet and Klein \cite{GK1}.  Given  $N \ge 2$,  we assume the induction hypothesis: Theorem~\ref{maintheorem} holds for $n=1,2,\ldots, N-1$ particles, and prove the theorem  for $N$ particles.  As in \cite{GK1}, the proof will be done by a bootstrapping argument, making successive use  of four multiscale analyses.

 \begin{induction} Let  $N \in \N $, $N\ge 2$, and  
   $E\up{N}> 0$.
For every  $\tau \in (0,  1)$ there is a length scale $L_{\tau} $,  $\delta_{\tau}>0$,  and $0<m_\tau^{*}  \le  \frac{1}{6}  \sqrt{E^{(N) }}$, 
such that  for $n = 1, 2, ..., N-1$ the following holds for all $E \le \En:=2^{N-n}E\up{N}$:
\begin{enumerate}
\item  For all $L \geq L_{\tau}$ and $\bolda \in  \R^{nd}$ we have 
 \begin{align} \label{eq001}
\P \Bigl\{   \nboxLa \,\, \text{is  $ \left (m^{*}_\tau, \,E \right )$-nonregular} \Bigr\}    \leq e^{-L^{\tau}}.
\end{align}

\item Fix $E < \En$ and let $I(E)=[E-\delta_{\tau}, E+\delta_{\tau}] \cap (-\infty, \En]$.  For all $L \geq L_{\tau}$  and all pairs  of $n$-particle boxes $\nboxLa$ and $\nboxLb$ with $d_{H} \pa{\bolda, \boldb} \geq L$, we get 
\begin{align} \notag
&\P \Bigl\{ \exists \, E^\pr \in I(E) \; \text{so both} \;  \nboxLa \text{ and}\:\nboxLb \; \text{are}\; \left ({m_\tau^*}, \,E^\pr \right )\text{-nonregular} \Bigr\} \\
& \hspace{50pt} \leq e^{-L^{\tau}} .  \label{eq002}
\end{align}
\end{enumerate}
 \end{induction}
 
Lemma \ref{PIsuit} (ii) will play an important role in the proof of   Theorem~\ref{maintheorem}.  To satisfy its hypotheses, the induction hypothesis specifies   $m_\tau^{*}  \le  \frac{1}{6}  \sqrt{E^{(N) }}$ for every  $\tau \in (0,  1)$, without loss of generality, and sets $\En:=2^{N-n}E\up{N}$. 

\smallskip
\textbf{In this section  we fix $N \in \N $, $N\ge 2$, and  an energy
   $E\up{N}> 0$, and  assume that   the induction hypothesis holds for this     $N$ and $E\up{N}$.} 

\smallskip

  For partially interactive $N$-particle boxes  we immediately get  probability estimates from the induction hypothesis.   
  
\begin{lemma} \label{PINSl}
Let $\Nboxlu = \boldlambda_{\ell}(\bold{u}_{\cJ}) \times  \boldlambda_{\ell}(\bold{u}_{\cJ^{c}} )$ be a PI $N$-particle box and $\tau \in (0,1)$. Then for $\ell$ large and all $E \leq \EN$  we have 
\begin{align} \notag
&\P  \Bigl\{ \Nboxlu \;\text{is} \; \left( m^*_{\tau}(\ell),E \right)\text{-nonregular} \Bigr\} \leq \ell^{Nd+1} e^{-\ell^{\tau}}\qtx {with}\\ 
& \qquad \qquad  m^*_{\tau}(\ell)= m_{\tau}^{*} - \tfrac{100 \pa{nd+1} \log \pa{2\ell}}{\ell}     ,\\ \notag
&\P  \Bigl\{ \Nboxlu \,\, \text{is} \,\, \left( \theta,\,E \right)\text{-nonsuitable} \Bigr\} \leq \ell^{Nd+1} e^{-\ell^{\tau}} \text{for} \; \; \theta < \tfrac \ell {\log \ell} \tfrac{ m^*_{\tau}}{100}, \\
&\P  \Bigl\{ \Nboxlu \,\, \text{is} \,\, \left(\tau,\,E \right)\text{-nonSES} \Bigr\} \leq \ell^{Nd+1} e^{-\ell^{\tau}}.  \notag
\end{align} 
\end{lemma}

\begin{proof} Let $E \leq \EN$.  It follows from Lemma \ref{PIsuit} (ii) and the induction hypothesis, using also  Lemma~\ref{NE},  that for large $\ell$,
\begin{align} 
&\P  \Bigl\{ \Nboxlu \; \text{is} \; \left( m^*_\tau(\ell),\,E \right)\text{-nonregular} \Bigr\}  \\ \notag
&\hskip30pt \leq  \sum_{\mu \in \sigma_{\cJ^c}  \cap (-\infty, \, 2\,\EN]    } \P  \Bigl\{  \boldlambda_{\ell}(\bold{u_{\cJ}})\text{ is }(m^*_\tau, E - \mu)\text{-nonregular}    \Bigr\} \\ 
& \hskip60pt  + \sum_{\lambda \in \sigma_{\cJ}  \cap (-\infty, \, 2\,\EN]  } \P  \Bigl\{  \boldlambda_{\ell}(\bold{u_{\cJ^{c}}}) \text{ is }(m^*_\tau, E - \lambda)\text{-nonregular}   \Bigr\} \notag \\
& \hskip30pt \le  C_{N, \, d} \, \pa{\EN }^{\tfrac{Nd}{2}}\ell^{Nd}e^{-\ell^{\tau}} \leq \ell^{Nd+1}e^{-\ell^{\tau}}. \notag
\end{align}
The other estimates now follow from Remark~\ref{goodbox}.
\end{proof}

In what follows, we  fix $\zeta, \, \tau, \beta, \, \zeta_0,\, \zeta_1,\, \zeta_2,\, \gamma$ such that 
\begin{gather} \label{constant}
0 < \zeta < \tau < 1,\quad  \zeta \, \gamma^{2} < \zeta_2,\\
0 < \zeta < \zeta_{2} < \gamma \zeta_2 <  \zeta_{1} < \gamma \zeta_1 < \beta < \zeta_0 < r < \tau < 1 \quad \text{with } \zeta \, \gamma^{2} < \zeta_2. \notag
\end{gather}
We set $m^*=  m_{\tau}^{*} $, where  $m_\tau^{*}  \le  \frac{1}{6}  \sqrt{E^{(N) }}$  is given   in the induction hypothesis.  

We will use  the Wegner estimates of   Theorem~\ref{Wegner0} and Corollary~\ref{Wegner2} for $n=1,2,\ldots,N$ particles, which  apply to an interval $I\subseteq [0,E_+)$   with $\abs{I}\le 2\gamma_{n,E_+}$.    In the multiscale analysis we will need   $E_+ = \En:=2^{N-n}E\up{N}$ for the $n$-particles Wegner estimates.  
For convenience, we take
$E_+= E\up{1}=2^{N-1}E\up{N}\ge \En$ for $n=1,2,\ldots,N$. Note that  the constants in these Wegner estimates (including $\gamma_{n,E_+}$) are increasing in $n$ and on $E_+$, so we will always take the constants for $n=N$ and  $E_+= E\up{1}$  (e.g., $\gamma_{N,E_+}$). To ensure that the condition $\abs{I}\le 2\gamma_{N,E_1}$ is always satisfied,
we will always  take sufficiently large scales $L$, i.e., $L \geq L(\gamma_{N, E_{1}})$, such that $L^{-s} \leq \gamma_{N, E_{1}}$ and $e^{-L^{\beta}} \leq \gamma_{N, E_{1}}$.  Moreover, in the following lemmas the conclusions are always assumed to hold for  $L$ sufficiently large.

 The proof of the induction step  proceeds as in  \cite{GK1,Kl}, with four multi-scale analyses, as in \cite{KlN}, using the toolkit  for the multiscale analysis in the continuum given in  Section~\ref{sectoolkilt}.  We  state all the steps, but refer to \cite{KlN} for the proofs when they are similar.

\subsection{The first multiscale analysis}
\begin{proposition} \label{part1mainthm} 
 Let $\theta > 8Nd$  and $E \leq \EN $.  Take $0 < p < p + Nd < s < s + 2Nd -2< \theta,$  $Y \ge 4000 N^{N+1} $, and $p_0=p_{0}(N) < \tfrac 12\pa{2Y}^{-Nd}$. Then there exists a length scale $Z_{0}^{*}$ such that if for some $L_0 \geq Z_{0}^{*}$ we have
\beq
\sup_{\boldx \in \R^{Nd}} \P \Bigl\{ \mathbf{\Lambda}_{L_0}^{(N)}(\boldx) \, \,  \text{is}\,\,(\theta, \,E)  \text{-nonsuitable} \Bigr\} \leq p_0, 
\eeq
then, setting $L_{k+1} = Y\,L_k,$ for $k = 0, 1, 2, ...,$ there exists $K_0 \in \N$ such that for every $k \geq K_0$ we have 
\beq
\sup_{\boldx \in \R^{Nd}} \P \Bigl\{ \mathbf{\Lambda}_{L_k}^{(N)} (\boldx)    \,\,\text{is}\,\,(\theta, \,E)  \text{-nonsuitable} \Bigr\} \leq L_{k}^{-p}.
\eeq
\end{proposition}

The proof of the  proposition uses the following deterministic lemma.  

\begin{lemma} \label{part1prop1a} 
Let $\theta > 8Nd$  and $E \leq \EN $.  Take $ Nd < s < s + 2Nd < \theta$. Let  $J \in \N$,   $Y \ge 4000J N^{N+1} $,  $L =  Y\ell$, and $\boldx \in \R^{Nd}$.
 Suppose  we have the following:
\begin{enumerate}
\item $\NboxLx$ is $E$-suitably nonresonant.

\item There are at most $J $ pairwise $\ell$-distant,  $(E, \, \theta)$-nonsuitable boxes in the $\ell$-suitable cover.

\item  Every box  $\mathbf{\Lambda}_{t}^{(N)}(\boldu)\subseteq \NboxLx$ with $t \in \set{ \pa{2  k_j \alpha  +1}\ell; j = 1 , \cdots, JN^N}$ and $\boldu \in \x + \alpha \ell \Z^{Nd}$, where $k_j$ is
 given  in Remark \ref{remcovering},  is $E$-suitably nonresonant.
\end{enumerate}
Then    the $N$-particle  box  $\NboxLx$ is $\left(E, \,\theta \right)$-suitable for $L$ sufficiently large.
\end{lemma}

Lemma~\ref{part1prop1a} has the same proof as  \cite[Lemma~3.3]{KlN}.  Prop~\ref{part1mainthm} is proved using Lemma~\ref{part1prop1a} as \cite[Proposition~ 3.2]{KlN} is proved using \cite[Lemma~3.3]{KlN}.

\subsection{The second multiscale analysis}
\begin{proposition} \label{part2mainthm}
 Let $E \le \EN$,  $p>0$, $\theta>0$, $1< \gamma <1 + \frac p{p+2Nd}$. Then  there exists a length scale $Z_{1}^{*}$ such that if for some $L_0 \geq Z_{1}^{*}$ we can verify
\beq
\sup_{\boldx \in \R^{Nd}} \P \Bigl\{ \mathbf{\boldlambda}_{L_0}^{(N)}(\boldx) \, \,  \text{is}\,\,(m_0, \,E) \text{-nonregular} \Bigr\} \leq L_{0}^{-p},
\eeq
where $ \theta \tfrac {\log {L_0}}{L_0}\le m_0<m^*$, 
then, setting $L_{k+1} = L_k^{\gamma},$ for $k = 1, 2, ...,$ we get
\beq
\sup_{\boldx \in \R^{Nd}} \P \Bigl\{ \mathbf{\boldlambda}_{L_k}^{(N)} (\boldx)    \,\,\text{is}\,\,\left(\tfrac{m_0}{2}, \,E \right) \text{-nonregular} \Bigr\} \leq L_{k}^{-p}\qtx{for all} k = 0, 1, 2, ...
\eeq
\end{proposition}

To prove the proposition we use the following deterministic lemma.

\begin{lemma} \label{part2prop1a}  
Let $E \leq \EN $, $L = \ell^{\gamma}$,  $J \in \N$, $m_0>0$, and 
\beq\label{mellcond}
m_\ell \in  [\tfrac{1}{\ell^{\kappa}},m_0], 
\qtx{where} 0< \kappa <\min \set { \gamma -1, \gamma (1-\beta), 1} . 
\eeq 
Suppose that we have the following:

\begin{enumerate}
\item $\NboxLx$ is $E$-nonresonant.

\item There are at most $J $ pairwise $\ell$-distant,  $(E, \, m_{\ell})$-nonregular boxes in the suitable cover.  

\item  Every box  $\mathbf{\Lambda}_{t}^{(N)}(\boldu)\subseteq \NboxLx$ with $t \in \set{ \pa{2  k_j \alpha  +1}\ell; j = 1 , \cdots, JN^N}$ and $\boldu \in \x + \alpha \ell \Z^{Nd}$, where $k_j$ is
 given  in Remark \ref{remcovering},  is $E$-nonresonant.

\end{enumerate}
Then $\NboxLx$ is $\left(E, \,m_L \right)$-regular for $L$ large, where 
\beq \label{mlmL}
 m_{\ell} \geq m_L \geq m_{\ell} -  \tfrac{1}{2\ell^{\kappa}}\ge  \tfrac{1}{L^{\kappa}}.
\eeq
\end{lemma}

Lemma~\ref{part2prop1a} and Proposition~\ref{part2mainthm} are proved in the same way as 
\cite[Lemma~3.5 and Proposition~ 3.4]{KlN}.

\subsection{The third multiscale analysis}

\begin{proposition}  \label{part3mainthm}
Let $E \leq \EN $,  $0 < \zeta_1 < \zeta_0 < 1$ as in  \eq{constant}, and  assume   $Y\ge \pa{3800N^{N+1}}^{\frac 1 {1-\zeta_0}}$. Then there exists $Z_2^{*} > L_{\tau}$ such that, if for some scale $L_0 > Z_2^{*}$ we have 
\beq   
\sup_{\x \in \R^{Nd}}\P \Bl\{ \boldlambda_{L_0}^{(N)}(\x) \text{ is } (\zeta_0, E)\text{-nonSES} \Br\}\le  \pa{2\pa{2Y}^{Nd}}^{-\tfrac {1}{Y^{\zeta_0} -1}}	,
\eeq
then, setting $L_{k+1}= Y\,L_k$, $k=0, 1, 2, ...,$ there exists $K_1 \in \mathbb{N}$ such that for every $k \geq K_1$ we have  
\beq 
\sup_{\x \in \R^{Nd}}\P \Bl\{\boldlambda _{L_k}(\boldx) \text{ is } (\zeta_0, E)\text{-nonSES}\Br\} \le  e^{-L_k^{\zeta_1}}.
\eeq
As a consequence, for every $k \geq K_1$, we have 
\beq 
\sup_{\x \in \R^{Nd}} \P \set{\boldlambda_{L_k}(\boldx) \text{ is } \left(L_k^{\zeta_0-1}, E\right)\text{-nonregular}} \le  e^{-L_k^{\zeta_1}}.
\eeq
\end{proposition}

The proof of  proposition uses the following deterministic lemma.

\begin{lemma} \label{part3prop1a}
Let $E \leq \EN $,  $L = Y \ell$, where   $Y\ge \pa{3800N^{N+1}}^{\frac 1 {1-\zeta_0}}$, and set $J = \lfloor  Y^{\zeta_0}\rfloor$, the largest integer $\le  Y^{\zeta_0}$.
 Suppose the following are true:
\begin{enumerate}
\item $\NboxLx$ is $E$-nonresonant.

\item There are at most J pairwise $\ell$-distant, $(E,\,\zeta_0)$-nonSES boxes in the suitable cover.

 \item  Every box  $\mathbf{\Lambda}_{t}^{(N)}(\boldu)\subseteq \NboxLx$ with $t \in \set{ \pa{2  k_j \alpha  +1}\ell; j = 1 , \cdots, JN^N}$ and $\boldu \in \x + \alpha \ell \Z^{Nd}$, where $k_j$ is
 given  in Remark \ref{remcovering},  is $E$-nonresonant.

\end{enumerate}

Then $\NboxLx$ is $(E,\,\zeta_0)$-SES, provided $\ell$ is sufficiently large.
\end{lemma}

Lemma~\ref{part3prop1a} and Proposition~\ref{part3mainthm} are proved in the same way as 
\cite[Lemma~3.7 and Proposition~ 3.6]{KlN}.

\subsection{The fourth multiscale analysis}

We fix $\zeta, \, \tau, \beta, \, \zeta_1,\, \zeta_2,\, \gamma$ as in \eq{constant}.

\subsubsection{The single energy multiscale analysis}
\begin{proposition} \label{part4mainthm1}

There  exists a length scale $Z_{3}^{*}$ such that, given an energy  $E \leq \EN $,   if for some $L_0 \geq Z_{3}^{*}$ we can verify
\begin{align} 
\sup_{\bolda \in \R^{Nd}}\P \Bigl\{   \mathbf{\Lambda}_{L_0}^{(N)}(\bolda)  \sqtx{is}  \left ({m_0}, \,E \right )\text{-nonregular} \Bigr\}    \leq e^{-L_0^{\zeta_1}},
\end{align}
where     $L_0^{\zeta_0-1} \le m_0< m^*$,  then for sufficiently large $L$ we have
\begin{align} 
\sup_{\bolda \in \R^{Nd}}\P \Bigl\{  \mathbf{\Lambda}_{L}^{(N)}(\bolda)  \sqtx{is}  \left (\tfrac {m_0}2, \,E \right )\text{-nonregular} \Bigr\}    \leq e^{-L^{\zeta_2}}.
\end{align}
\end{proposition}

Proposition~\ref{part4mainthm1} is proved first  for a sequence of length scale $L_k$ similarly  to   Proposition~\ref{part2mainthm};  to obtain the sub-exponential decay of probabilities we choose $J$, the number of bad boxes, dependent on the scale $L$ as in the proof of Proposition~\ref{part4mainthm} below. To obtain Proposition~\ref{part4mainthm1} as stated, that is, for all sufficiently large scales, we prove a slightly more general result.

\begin{definition}
Let $E \in \R$.   An N-particle box, $\NboxLx$, is said to be $(E, m_{L})$-good if and only if 
it is $(E, m_{L})$-regular and $E$-nonresonant.
\end{definition}

\begin{lemma} \label{GKlemma} Let $\NboxLx$ be an N-particle box, $\gamma > 1$,  $\ell=L^{\frac 1 {\gamma^\pr}} $ with $\gamma \le \gamma^\pr \le \gamma^2$, and $m>0$.
Let $E \leq \EN $, and suppose every box in the suitable cover is $(E, m)$-good.   Then $\NboxLx$ is $\pa{E, \frac m 2}$-good for large $L$.          
\end{lemma}

This lemma is just \cite[Lemma~3.16]{GKber}.

\begin{lemma} \label{part1thm}
Let $E_1 \le \EN$, $\zeta_2 \in (\zeta, \, \tau)$, and $\gamma \in (1, \, \tfrac{1}{{\zeta_2}})$  with $\zeta \, \gamma^{2} < \zeta_2  $.   Assume there exists a mass $m_{\zeta_2}>0$ and  a length scale $L_0 = L_{0}(\zeta_2)$, such that,  taking $L_{k+1} = L_{k}^{\gamma}$ for $k=0,1,\ldots$,  we have 
\begin{align} \label{appeneq1}
\sup_{\bolda\in \R^{Nd}} P \Bigl\{ \mathbf{\Lambda}_{L_{k}}^{(N)}(\bolda)\sqtx{is not} \left ({m_{\zeta_2}}, \,E_1 \right)\text{-good}  \Bigr\}  \leq e^{-L_{k}^{\zeta_2}} \qtx{for} k=0,1,\ldots.\end{align}
Then there exists $L_{\zeta}$ such that for every $L \geq L_{\zeta}$ we have 
\begin{align}
\sup_{\bolda\in \R^{Nd}} P \Bigl\{ \mathbf{\Lambda}_{L}^{(N)}(\bolda)\sqtx{is not} \left ({m_{\zeta_2}}, \,E_1 \right)\text{-good}  \Bigr\}  \leq e^{-L^{\zeta}}.
\end{align}
\end{lemma}

The proof of Lemma \ref{part1thm} is straightforward (see \cite[Lemma 3.11]{KlN}).

\subsubsection{The energy interval multiscale analysis}

\begin{lemma} \label{part4lem0}
Let $\NboxLx$ be an N-particle box and  $m > 0$. Let  $E_0 \le \EN$, and suppose that

\begin{enumerate}
\item $\NboxLx$ is $(m,E_0)$-regular,
\item $ \dist \Bl( \sigma\left (H_{\NboxLx}  \right) , E_0 \Br) \geq e^{-L^{\beta}}$, i.e.,  $\norm{R_{\NboxLx}(E_0)} \leq e^{L^{\beta}}$.
\end{enumerate}
Then $\NboxLx$ is $\left(m - \frac {100\log 2} L,E  \right)$-good for every 
$E \in I = \left( E_0 - \eta, E_0 + \eta \right)$, where $\eta = \tfrac{1}{2} e^{-mL- 2L^\beta}$.
\end{lemma}

Lemma~\ref{part4lem0} is proved as \cite[Lemma~3.12]{KlN}.

Proposition~\ref{part3mainthm}, combined with Theorem~\ref{Wegner0}  and Lemma \ref{part4lem0},    yields the following proposition.

\begin{proposition} \label{bridgethm}
Let $0 < \zeta_2 < \zeta_1 < \zeta_0 < 1$, and assume the conclusions of Proposition~\ref{part3mainthm}.   There exists scales $L_k$, $k=1,2,\ldots$, such that $\lim_{k\to \infty} L_k=\infty$, with the following property: Let 
\beq
m_k=\left(L_k^{\zeta_0-1} - \tfrac {100\log 2} {L_k}  \right) \qtx{and}
\eta_k= \tfrac{1}{2} e^{- L_k^{\zeta_0}- 2L_k^\beta}.
\eeq
 Then for all  $E_0 \leq \EN$ we have
 \begin{align}\notag
& \sup_{\x \in \R^{Nd}} \P \set{\exists E\in \left( E_0 - \eta_k, E_0 + \eta_k \right) \sqtx{such that} \boldlambda_{L_k}(\boldx) \sqtx{is} \left(m_k, E\right)\text{-nonregular}} \\
 & \hspace{50pt} \le  e^{-L_k^{\zeta_1}},
\end{align}
and
\begin{align}\notag
&\sup_{\x \in \R^{Nd}} \P \set{\exists E\in \left( E_0 - \eta_k, E_0 + \eta_k \right) \sqtx{such that} \boldlambda_{L_k}(\boldx) \text{ is not } \left(m_k, E\right)\text{-good}} \\
& \hskip50pt \le   e^{-L_k^{\zeta_2}}.
\end{align}
\end{proposition}

 We now  take $L=\ell^\gamma$.

\begin{definition}
Let $\NboxLx = \boldlambda_{L}(\bold{x}_{\cJ}) \times  \boldlambda_{L}(\bold{x}_{\cJ^{c}} )$ be a PI N-particle box with the usual $\ell$ suitable cover, and consider an energy $E \in \R$. Then:

\begin{enumerate}
\item $\NboxLx$ is not $E$-Lregular \emph(for left regular\emph) if and only if there are two boxes in the suitable cover of $\boldlambda_{L}(\bold{x}_{\cJ} )$
that are $\ell$-distant and $(m^{*}, \, E-\mu)$-nonregular  
for some $\mu \in \sigma \left( H_{ \boldlambda_{L}(\bold{x}_{\cJ^{c}} )  }  \right) \cap (-\infty, 2\EN]$. 

\item $\NboxLx$ is not $E$-Rregular \emph(for right regular\emph) if and only if there are two boxes in the suitable cover of $\boldlambda_{L}(\bold{x}_{\cJ^c})$   that are $\ell$-distant and $(m^{*}, \, E-\lambda)$-nonregular for some $\lambda \in \sigma \left( H_{ \boldlambda_{L}(\bold{x}_{\cJ} )  }  \right) \cap (-\infty, 2\EN] $.

\item $\NboxLx$ is  $E$-preregular if and only if  $\NboxLx$ is  $E$-Lregular and $E$-Rregular. 
\end{enumerate}
\end{definition}

\begin{lemma} \label{prereg}
Let $E_0 \leq \EN$ such that $I = [E_0-\delta_\tau, \, E_0+\delta_\tau] \subseteq (-\infty, \,  2\EN]$,  and consider a PI N-particle box $\NboxLx = \boldlambda_{L}(\bold{x}_{\cJ}) \times  \boldlambda_{L}(\bold{x}_{\cJ^{c}} )$.
Then 
\begin{enumerate}
\item $\P\set{\NboxLx \text{ is not }E\text{-Lregular for some } E \in I} \leq L^{3Nd}  e^{-\ell^{\tau}}$, 
\item $\set{\NboxLx \text{ is not }E\text{-Lregular for some } E \in I} \leq L^{3Nd} e^{-\ell^{\tau}}$.

\end{enumerate}
We conclude that for $L$ sufficiently large we have 
\item  \beq
\P \bigl\{ \NboxLx \text{ is not }E\text{-preregular for some } E \in I  \bigr\} \leq 2L^{3Nd}  e^{-\ell^{\tau}} .
\eeq
\end{lemma}

Lemma~\ref{prereg} has the same proof as \cite[Lemma~3.15]{KlN}.

\begin{definition} \label{CNR}
Let $\NboxLx = \boldlambda_{L}(\bold{x}_{\cJ}) \times  \boldlambda_{L}(\bold{x}_{\cJ^{c}} ) $ be a PI N-particle box, and consider an energy $E \leq \EN$. Then:

\begin{enumerate}
\item $\NboxLx$ is $E$-left nonresonant \emph(or LNR\emph) if and only if every box $\boldlambda_{\pa{2k_j +1 }\ell}(\bolda)$ with $\boldlambda_{\pa{2k_j +1 }\ell}(\bolda) \subseteq \boldlambda_{L}(\boldx_{\cJ})$, $\bolda \in \x_{\cJ} + \alpha \ell \Z^{\abs{\cJ}d}$ and $j \in \set{1, \, 2, \ldots \abs{\cJ}^{\abs{\cJ}}}$, is $(E-\mu)$-nonreso\-nant for every 
$\mu \in \sigma \pa{H_{\boldlambda_{L}(\boldu_{\cJ^c})}} \cap (-\infty, 2\EN]$. Otherwise we say $\NboxLx$ is $E$-left resonant \emph(or LR\emph).

\item $\NboxLx$ is $E$-right nonresonant \emph(or RNR\emph) if and only if for every box
$\boldlambda_{\pa{2k_j +1 }\ell}(\bolda) \subseteq \boldlambda_{L}(\boldx_{\cJ^{c}})$ with $\bolda \in \x_{\cJ^c} + \alpha \ell \Z^{\abs{\cJ^c}d}$ and $j \in \set{1, \, 2, \ldots \abs{\cJ^c}^{\abs{\cJ^c}}}$ is $(E-\lambda)$-nonreso\-nant for every 
$\lambda \in \sigma \pa{H_{\boldlambda_{L}(\boldx_{\cJ})}} \cap (-\infty, 2\EN]$. Otherwise we say $\NboxLx$ is $E$-right resonant \emph(or RR\emph).

\item We say $\NboxLx$ is
$E$-highly nonresonant \emph(or HNR\emph) if and only if $\NboxLx$ is $E$-nonresonant, 
$E$-LNR, and $E$-RNR.
\end{enumerate} 
\end{definition}

\begin{lemma}\label{part2firstthm}
Let $E \leq \EN$ and 
$\NboxLx = \boldlambda_{L}(\bold{x}_{\cJ}) \times  \boldlambda_{L}(\bold{x}_{\cJ^{c}} )$ 
be a PI N-particle box. Assume that the following are true:
\begin{enumerate}
\item $\NboxLx$ is $E$-HNR.
\item $\NboxLx$ is $E$-preregular.
\end{enumerate}
Then $\NboxLx$ is $\left(m(L), \, E \right)$-regular for sufficiently large $L$, where 
 \beq   \label{mL}
m(L) =m^* - \tfrac{1}{2L^{\kappa}} - \tfrac{100 \pa{nd+1} \log \pa{2L}}{L}.
 \eeq
\end{lemma}

\begin{proof}
Applying Lemma \ref{PIsuit}(ii), it is sufficient to prove that there exists $m \le \frac{1}{6}  \sqrt{E^{(N) }}$ such that $\boldlambda_{L}(\bold{x_{\cJ}})$ is $(m,\, E - \mu)$-regular  for every $\mu \in \sigma_{\cJ^c} \cap (-\infty, \,2\,E^{(N)}]$ and  $\boldlambda_{L}(\bold{x_{\cJ^{c}}})$ is $(m,\, E - \lambda)$-regular  for every $\lambda \in \sigma_{\cJ} \cap (-\infty, \,2\,E^{(N)}]$. Then we can conclude that $\NboxLx$ is $\pa{ m - \tfrac{100 \pa{nd+1} \log \pa{2L}}{L}, \, E }$-regular.

Let $\mu \in \sigma_{\cJ^c} \cap (-\infty, \,2\,E^{(N)}]$. Since $\NboxLu$ is $E$-preregular; thus it is $E$-Lregular, which implies there cannot be two boxes in the suitable cover of $\blam_{L}(\x_{\cJ})$ that are $\ell$-distant and $\pa{m^*, E-\mu}$-nonregular. Moreover, $\NboxLu$ is $E$-HNR; thus it is $E$-LNR, which implies every box $\boldlambda_{\pa{2k_j +1 }\ell}(\bolda)$ with $\boldlambda_{\pa{2k_j +1 }\ell}(\bolda) \subseteq \boldlambda_{L}(\boldx_{\cJ})$, $\bolda \in \x_{\cJ} + \alpha \ell \Z^{\abs{\cJ}d}$ and $j \in \set{1, \, 2, \ldots \abs{\cJ}^{\abs{\cJ}}}$, is $(E-\mu)$-nonreso\-nant. By Lemma \ref{part2prop1a}, we have $\boldlambda_{L}(\bold{x_{\cJ}})$ is $\pa{m^* - \tfrac{1}{2L^{\kappa}},\, E - \mu}$-regular. A similar argument will apply for $\lambda \in \sigma_{\cJ} \cap (-\infty, \,2\,E^{(N)}]$. 
We conclude that   $\NboxLx$ is $\pa{ m^* - \tfrac{1}{2L^{\kappa}} - \tfrac{100 \pa{nd+1} \log \pa{2L}}{L}, \, E }$-regular.
\end{proof}

\begin{lemma}
Let $E \leq \EN,$ and $\NboxLx = \boldlambda_{L}(\bold{x}_{\cJ}) \times  \boldlambda_{L}(\bold{x}_{\cJ^{c}} ) $ be a PI N-particle box.
\begin{enumerate}
\item If $\NboxLx$ is  E-right resonant, then there exists an N-particle rectangle
\beq   
\mathbf{\boldlambda} = \boldlambda_{L}(\boldx_{\cJ}) \times  \boldlambda_{\pa{2k_j\alpha+1}\ell} (\boldu) ,
\eeq
where $j \in \set{1, 2, \ldots, \abs{\cJ^c}^{\abs{\cJ^c}} } $,  $\boldu \in \x_{\cJ^c} + \alpha \ell \Z^{\abs{\cJ^c}d}$, and $  \boldlambda_{\pa{2k_j\alpha+1}\ell} (\boldu)   \subseteq \boldlambda_{L}(\boldx_{\cJ^{c}}) $,  such that 
\beq 
\dist \bigl( \sigma \left(  \mathbf{\boldlambda} \bigr) , E    \right) < \tfrac{1}{2} e^{-\pa{\pa{2k_j\alpha+1}\ell}^{\beta}} \leq \tfrac{1}{2} e^{-\ell^{\beta}}.
\eeq

\item If $\NboxLx$ is E-left resonant, then there exists an N-particle rectangle
\beq   
\mathbf{\boldlambda} = \boldlambda_{\pa{2k_j\alpha+1}\ell} (\boldu)   \times \boldlambda_{L}(\boldx_{\cJ^{c}}) ,
\eeq
where  $j \in \set{1, 2, \ldots, \abs{\cJ}^{\abs{\cJ}} }$,  $\boldu \in \x_{\cJ} + \alpha \ell \Z^{\abs{\cJ}d}$, and $  \boldlambda_{\pa{2k_j\alpha+1}\ell} (\boldu)    \subseteq \boldlambda_{L}(\boldu_{\cJ}) $,   such that 
\beq 
\dist \bigl( \sigma \left(  \mathbf{\boldlambda} \bigr) , E    \right) < \tfrac{1}{2} e^{-\pa{\pa{2k_j\alpha+1}\ell}^{\beta}} \leq \tfrac{1}{2} e^{-\ell^{\beta}}.
\eeq
\end{enumerate}
\end{lemma} 

\begin{proof}
Let $E \leq \EN$ and $\NboxLx = \boldlambda_{L}(\bold{x}_{\cJ}) \times  \boldlambda_{L}(\bold{x}_{\cJ^{c}} ) $ be a PI N-particle box.  Suppose $\NboxLx$ is $E$-right resonant. (The same argument applies if  $\NboxLx$ is $E$-left resonant.) Then we can find $\lambda \in \sigma \left( H_{\boldlambda_{L}(\boldx_{\cJ})}  \right) \cap (-\infty, 2\EN]$ and an $ \bigl ( N- \abs{\cJ} \bigr)$-particle box, $\boldlambda_{\pa{2k_j \alpha+1}\ell} (\boldu) \subseteq \boldlambda_{L}(\boldx_{\cJ^{c}})$, with $\boldu \in \x_{\cJ^c} + \alpha \ell \Z^{Nd}$ and  $j \in \set{1, 2, \ldots, \abs{\cJ^c}^{\abs{\cJ^c}} }$, such that $\boldlambda_{\pa{2k_j \alpha+1}\ell} (\boldu)$ is $(E-\lambda)$-resonant, so 
there exists $\eta \in \sigma \bigl( H_{ \boldlambda_{\pa{2k_j \alpha+1}\ell} (\boldx) } \bigr)$ such that 
$ 
\abs{ E - \lambda - \eta }< \tfrac{1}{2} e^{-\pa{\pa{2k_j \alpha+1}\ell}^{\beta}}.
$
Moreover, 
$ \boldlambda_{L}(\bold{x}_{\cJ}) \times  \boldlambda_{L}(\bold{x}_{\cJ^{c}} )$ is PI and  $  \boldlambda_{\pa{2k_j \alpha+1}\ell} (\boldu) \subseteq \boldlambda_{L}(\boldx_{\cJ^{c}})$, so if we take $\mathbf{\boldlambda} = \boldlambda_{L}(\boldx_{\cJ}) \times  \boldlambda_{\pa{2k_j \alpha+1}\ell} (\boldu)$ we get
\beq 
\sigma \left( H_{ \mathbf{\boldlambda} } \right)  = \sigma \Bl( H_{\boldlambda_{L}(\boldx_{\cJ}) }  \Br) + \sigma \left( H_{  \boldlambda_{\pa{2k_j \alpha+1}\ell} (\boldu)} \right).
\eeq
Hence, if a PI N-particle box $\NboxLx = \boldlambda_{L}(\bold{x}_{\cJ}) \times  \boldlambda_{L}(\bold{x}_{\cJ^{c}} ) $ is $E$-right resonant, then there exists an N-particle box
$\mathbf{\boldlambda} = \boldlambda_{L}(\boldx_{\cJ}) \times  \boldlambda_{\pa{2k_j \alpha+1}\ell} (\boldu),$ 
where $  \boldlambda_{\pa{2k_j \alpha+1}\ell} (\boldu) \subseteq \boldlambda_{L}(\boldx_{\cJ^{c}})$,       such that 
\begin{align*}
\dist \bigl( \sigma \left( H_{ \mathbf{\boldlambda} } \bigr) , E    \right) < \tfrac{1}{2}e^{-\pa{\pa{2k_j \alpha+1}\ell}^{\beta}}.
\end{align*}
\end{proof}

We now state the energy interval multiscale analysis.
Given $m > 0$,  $L \in \N$, $\x, \, \y \in \Ndspace$, and an interval $I$, we define the event
\begin{align}
&R \left( m, \, I,\, \x, \, \y, \, L, \, N \right) =   \notag\\
& \qquad \left\{ \exists \, E \in I \sqtx{such that}\boldlambda _{L}^{(N)}(\boldx) \text{ and } \boldlambda _{L}^{(N)}(\boldy)  \text{ are not } \left(m, E\right)\text{-regular}  \right\} . 
\end{align}

\begin{proposition} \label{part4mainthm}.
 Let $\zeta, \, \tau, \beta, \, \zeta_1,\, \zeta_2,\, \gamma$ as in \eq{constant} and   $0<m_0 <m^*$. 
 There exists a length scale $Z_{3}^{*}$ such that, given a closed  interval    $I \subseteq (-\infty,E\up{N}]$,  if for some $L_0 \geq Z_{3}^{*}$ we can verify
\beq
\P \Bigl\{R \left( m_0, \, I,\, \x, \, \y, \, L_0, \, N \right)\Bigr\}  \leq e^{-L_0^{\zeta_2}},
\eeq
for every pair of partially separated $N$-particle boxes $\boldlambda_{L_0}^{N}(\x)$ and $\boldlambda_{L_0}^{(N)}(\y)$, then, setting   $L_{k+1} = L_k^{\gamma} = L_0^{\gamma^k}$  for  $ k=0,1,2,\ldots$, 
 for    every pair of partially separated $N$-particle  boxes $\boldlambda_{L_k}^{(N)}(\boldx)$ and $\boldlambda_{L_k}^{(N)}(\boldy)$ we have 
\begin{align}
&\P \Bigl\{R \left( \tfrac{m_0}{2}, \, I,\, \x, \, \y, \, L_k, \, N \right) \Bigr\}  \\ \notag
&\qquad  \le  \P \Bigl\{ \exists \, E \in I \,\,so \,\, \mathbf{\Lambda}_{L_k}^{(N)}(\boldx) \,\, and \,\, \mathbf{\Lambda}_{L_k}^{(N)}(\boldy)  \,\,are\,\, not \,\, \left ( \tfrac{m_0}{2}, \,E \right )\text{-good} \Bigr\}  \leq e^{-L_k^{\zeta_2}}.
\end{align}
\end{proposition}

Proposition~\ref{part4mainthm} is proved in the same way as \cite[Propositions~3.19 and 3.21]{KlN}. The dependence of  the length scale  $Z_{3}^{*}$ on $E\up{N}$ (not present in \cite{KlN}) comes from the use of Theorem~\ref{Wegner0} and  Corollary~\ref{Wegner2}.

\subsection{Completing the proof of the bootstrap multiscale analysis}
Proceeding as in \cite[Section~6]{GK1}, Theorem~\ref{maintheorem} follows from Propositions~\ref{part1mainthm}, \ref{part2mainthm}, \ref{part3mainthm},  plus Proposition~\ref{part4mainthm1} for Part (i) (the single energy bootstrap multiscale analysis), and Propositions~\ref{bridgethm} and  \ref{part4mainthm} (the energy interval bootstrap multiscale analysis). 

\section{From the bootstrap multiscale analysis to localization}\label{secMSAloc}

Corollary \ref{localization} is proved from Theorem~\ref{maintheorem}  along the lines of  the proofs of the corresponding statements in  \cite{vDK, GK1, GKjsp,GKber}), similarly to the proof of \cite[Corollary~1.7]{KlN} from \cite[Theorem~1.5]{KlN}.

\appendix

\section{The almost-sure spectrum of the $n$-particle Anderson Hamiltonian}\label{apbottom}

\begin{proposition}
Let $\Sigma^{(n)}$ be the almost-sure spectrum of  the n-particle Anderson Hamiltonian   $H_{\bom}^{(n)}$  as in Definition~\ref{defAndmodel}.  Then $\Sigma^{(n)}=[0,\infty)$
\end{proposition}

\begin{proof}  Clearly $\Sigma^{(n)}\subset [0,\infty)$. We need to prove $[0,\infty)\subset \Sigma^{(n)}$.  So let $\lambda \in [0,\infty)$ and $\eps >0$. Since $\sigma (- \Delta^{(n)})=[0,\infty)$,  there exists $ \psi \in C^{2}(\R^{nd})$ with $\norm{\psi}=1$ and a box  $\nboxLx$ such that $\supp \psi \subseteq \nboxLx$ and $\norm{(-\Delta^{(n)}-\lambda)\psi} \leq \tfrac{\epsilon}{2}$.   Without loss of generality we can assume that $U=0$ on $\nboxLx$,
so $H_{\bom}^{(n)}\psi= H_{0,\bom}^{(n)}\psi$.  If  we have  $V_{\bom}^{(n)} \Chi_{\nboxLx} \leq \tfrac{\epsilon}{2} $,  we conclude that $\norm{(H_{\bom}^{(n)}-\lambda)\psi} \leq \eps$, and hence  $\dist \pa{\lambda, \sigma \pa{H_{\bom}^{(n)}}}\le \eps$. Thus,
\beq\label{almosteps}
\P \set{\dist \pa{\lambda, \sigma \pa{H_{\bom}^{(n)}}}\le \eps}  \supseteq \P \set{ V_{\bom}^{(n)} \Chi_{\nboxLx} \leq \tfrac{\epsilon}{2}  } >0,
\eeq
where the strict positivity comes from Definition~\ref{defAndmodel}.

Since $\Sigma^{(n)}=\sigma \pa{H_{\bom}^{(n)}}$ for $\P$-a.e.\ $\bom$, we conclude from  \eq{almosteps} that we have   $\dist \pa{\lambda, \Sigma^{(n)}}\le \eps$ for all $\eps>0$, and hence $\lambda \in \Sigma^{(n)}$.
\end{proof}

\section{Unique continuation principle for spectral projections of Schr\" odinger operators on arbitrary rectangles}  \label{apUCP}

In this appendix we extend \cite[Theorems~1.1 and 2.2]{Kl2} to arbitrary rectangles.
Let $H= -\Delta + V$ be a  Schr\" odinger operator  on  $\mathrm{L}^2(\R^d)$.  Given a rectangle $\La \subset \R^d$,  let $H_\La= -\Delta_\La+ V_\La$ denote the  restriction of $H$   to    the  rectangle $\La$ with either Dirichlet or periodic boundary condition:    $\Delta_\La$ is the Laplacian  with either Dirichlet or periodic boundary condition and $V_\La$ is the restriction of $V$ to $\La$.  (We will abuse the notation and simply write $V$ for $V_\La$, i.e., $H_\La= -\Delta_\La+ V$ on $\L^2(\La)$.)
By  a unique continuation principle for  spectral projections  (UCPSP)   we mean an estimate
of the form
\beq\label{UCPSP}
\Chi_I (H_\La)  W \Chi_I (H_\La) \ge \kappa\,   \Chi_I (H_\La),
\eeq 
where $\Chi_I$ is the characteristic function of  an interval $I \subset \R $, $W\ge 0 $ is a  potential, and $\kappa >0$ is a constant.

In this appendix we use the Euclidean norm on  $\R^d$: 
\beq
\abs{x}=\abs{x}_2:=\pa{\sum_{j=1}^{d} \abs{x_{j}}^{2}}^{\frac 12 } \qtx{for} x=\pa{x_1,x_2,\ldots, x_d} \in \R^d.
\eeq
Distances between sets in $\R^{d}$ will be measured with respect to norm $\abs{x}$. The ball centered at $x\in \R^d$ with radius $\delta>0$ is given by
\beq
B(x,\delta):= \set{y \in \R^{d}; \, \abs{y-x}<\delta}. 
\eeq  
  We consider rectangles 
\beq\label{rect}
\La= \La_{\bL}(a)= a + \prod_{j=1}^d  \,  (-\tfrac {L_j} 2, \tfrac {L_j} 2)=  \prod_{j=1}^d  \,  (a_j -\tfrac {L_j} 2, a_j +\tfrac {L_j} 2), \eeq
where $a \in \R^d$ and $\bL=(L_1,\ldots,L_d)\in (0,\infty)^d $.  The  box $\La_L(x)= x + (-\tfrac L 2, \tfrac L 2)^d$  centered at $x\in \R^d$ with side of length  $L  $ is the special case  $L_1=\ldots=L_d=L$.  Given a rectangle $\La$ we set
\
\beq
\widehat{\La}=\La \cap \Z^d\qtx{and} \widehat{\widehat{\La}}=\set{k \in \widehat{\La}; \Lambda_1(k)\subset \La}.
\eeq
$H_\La$ will denote the restriction of $H$ to the rectangle $\La$ with either Dirichlet or periodic boundary condition.

 Given subsets $A$ and $B$ of $\R^{d}$, and  a  function $\vphi$ on  the set $B$,  we set $\vphi_{A}:=\vphi \Chi_{A\cap B}$.  In particular, given $x\in   \R^{d}$ and $\delta >0$ we write $\vphi_{x,\delta}: =\vphi_{B(x, \delta )}$.   We let $\N_{\mathrm{odd}}$  denote the set of odd natural numbers.
If $K$ is an operator on a Hilbert space, $\cD(K)$ will denote its domain.   By a  constant  we will always mean a finite constant.  We will use  $C_{a,b, \ldots}$, $C^{\pr}_{a,b, \ldots}$,  $C(a,b, \ldots)$, etc., to  denote a constant depending only on the parameters
$a,b, \ldots$. 

The following is an extension of \cite[Theorem~1.1]{Kl2} to rectangles with arbitrary centers and side lengths.

\begin{theorem}\label{thmUCPSP} Let  $H =   -\Delta + V$ be a  Schr\"odinger operator on 
$\mathrm{L}^2(\mathbb{R}^d)$,   where $V$ is a bounded potential.  Fix  $\delta \in (0,\frac 1 {2}]$, let  $\set{y_k}_{k \in \Z^d}$ be sites in $\R^d$ with  $B(y_k,\delta) \subset \La_1(k)$ for all $k\in \Z^d$.
Given $E_{0}>0$, set $K =K(V,E_0)= 2\norm{V}_{\infty}+ E_{0} $.  Consider  a rectangle $\La$ as in \eq{rect}, where $a\in \R^d$ and $L_j \ge 114 \sqrt{d}$ for $j=1,\ldots,d$, and set
\beq\label{defWthm}
W\up{\La}= \sum_{k \in \widehat{\widehat{\La}} } \Chi_{B(y_k,\delta) }.
\eeq
 There exists a constant $M_d>0$, such that, defining $\gamma =\gamma(d,
K,\delta) >0$ by
\beq\label{defgamma}
\gamma^2=\tfrac 1 2   \delta ^{M_d \pa{1 + K^{\frac 2 3}}} ,
\eeq
 then for any closed  interval $I \subset (-\infty, E_{0}]$ with $\abs{I} \le 2\gamma $  we have
\beq \label{chivchi}
\Chi_{I}(H_\La) W\up{\La} \Chi_{I}(H_\La) \ge  \gamma^2  \Chi_{I}(H_\La).
\eeq
\end{theorem}

\begin{remark} It follows,  using Theorem~\ref{thmUCPSP} in the proofs,  that the optimal Wegner estimates for 
(one-particle) crooked Anderson Hamiltonians given in \cite[Theorems~1.4 and 1.5]{Kl2} hold for a rectangle $\La$ as in \eq{rect}, where $a\in \R^d$ and $L_j \ge 114 \sqrt{d}+\delta_+$ for $j=1,\ldots,d$. (In particular, they hold on arbitrary
boxes  $\La=\La_L(x_0)$, where $x_0 \in \R^d$ and $L \ge 114\sqrt{d}+ \delta_+$.)
\end{remark}

For convenience we recall the quantitative unique continuation principle \cite[Theorem~3.2]{BKl} as stated in \cite[Theorem~2.1]{Kl2}.

\begin{theorem}\label{thmucp} Let  ${\Omega}$ be an  open subset  of $\R^d$ and consider    a real measurable function $V$ on ${\Omega}$ with $\norm{V}_{\infty} \le K <\infty$.   Let 
$\psi \in\mathrm{H}^2({\Omega})$ be real valued and  let  ${\zeta} \in \L^2({\Omega})$  be defined by
\beq \label{eq}
-\Delta {\psi} +V{\psi}={\zeta}  \qtx{a.e.\  on} \Omega.
\eeq
 Let   ${\Theta} \subset {\Omega}$  be a bounded measurable set where $\norm{\psi_{\Theta}}_2 >0$. 
Set
\beq \label{defRx0}
{Q}(x,\Theta):= \sup_{y \in \Theta } \abs{y - x} \qtx{for} x \in {\Omega}.
\eeq
Consider $x_0 \in {\Omega}\setminus \overline{\Theta}$ such that
\beq
  \label{xR}
{Q}={Q}(x_0,\Theta)\ge  1 \qtx{and} B(x_0, 6{Q}+ 2)\subset {\Omega}.
\eeq
Then, given
\beq \label{delta}
0<  \delta \le \min\set{   \dist \pa{x_0, {\Theta}},\tfrac 1 {2}},
\eeq
we have
\begin{align} \label{UCPbound}
 \pa{\frac \delta{Q}}^{m_d \pa{1 + K^{\frac 2 3}}\pa{ {Q}^{\frac 43}  +  \log \frac{\norm{ {\psi}_{{\Omega}}}_{2}} {\norm{ {\psi}_{{\Theta}}}_2}}}\norm{ {\psi}_{{\Theta}}}^2_2  \le   \norm{ {\psi}_{x_0,\delta}}^2_2 + \delta^2 \norm{{\zeta_{{\Omega}}}}_2^2,
\end{align}
where $m_d>0$ is a constant depending only on $d$.
\end{theorem}

The following theorem is a version of \cite[Theorem~2.2]{Kl2} for rectangles with arbitrary centers and side lengths.

\begin{theorem} \label{lemUCPeig} Let  $H =   -\Delta + V$ be a  Schr\"odinger operator on 
$\mathrm{L}^2(\mathbb{R}^d)$, where V is a bounded potential with $\norm{V}_{\infty}\le K$.  Fix  $\delta \in (0,\frac 1 {2}]$, let  $\set{y_k}_{k \in \Z^d}$ be sites in $\R^d$ with  $B(y_k,\delta) \subset \La_1(k)$ for all $k\in \Z^d$. 
 Consider a rectangle $\La$ as in \eq{rect}, where $a\in \R^d$ and $L_j \ge 114 \sqrt{d}$ for $j=1,\ldots,d$.   Then for all  real-valued $\psi \in \cD(\Delta_\La) $    
we have
\begin{align} \label{UCPdelta}
 \delta^{M_d \pa{1 + K^{\frac 2 3}}} \norm{\psi_\Lambda}_2^2   \le  
\sum_{k \in \widehat{\widehat{\La}}}\norm{ {\psi}_{y_{k},\delta}}^2_2 +  \delta^2 \norm{\pa{(-\Delta + V)\psi}_\Lambda}_2^2,
\end{align}
where $M_d>0$ is a constant depending only on $d$.
\end{theorem}

\begin{proof}
As in \cite[Proof of Corollary~A.2]{GKber}, we extend $V$ and  functions $\vphi \in \mathrm{L}^2(\Lambda)$ to $\R^d$.

For 
Dirichlet boundary condition, 
given $\vphi \in \mathrm{L}^2(\Lambda)$, we extend it to a function $\widetilde{\vphi}\in \mathrm{L}^2_{\mathrm{loc}}(\R^d)$ by setting $\widetilde{\vphi}=\vphi$ on $\Lambda$ and $\widetilde{\vphi}=0 $ on $\partial \Lambda$, and requiring 
\beq \label{widetildephi}
\widetilde{\vphi}(x)=- \widetilde{\vphi}(x + (L_j  -2 \theta_j(x_j-a_j)) \e_j)\sqtx{for all} x \in \R^d \sqtx{and} j \in \set{1,2 \dots,d},
\eeq
where $\set{\e_j}_{j =1,2\ldots,d}$ is the canonical orthonormal basis in $\R^d$, and for each  $t\in  \R$ we define $\theta_j(t)\in  ( -\frac {L_j} 2, \frac {L_j}  2]$ by $t =
kL_j + \theta_j(t)$ with $k \in \Z$.  We also  extend the potential $V$ to a potential $\widehat{V}$ on  $\R^d$ by  by setting $\widehat{V}=V$ on $\Lambda$ and $V=0$ on $\partial \Lambda$, and requiring that 
for all $x \in \R^d$ and $j \in \set{1,2 \dots,d}$ we have
\beq
\widehat{V}(x)=\widehat{V}(x + (L_j  -2 \theta_j(x_j-a_j)) \e_j).
\eeq
Note that  $\|\widehat{V}\|_\infty =\norm{V}_\infty \le K$.  Moreover, 
$\psi \in \cD(\Delta_\Lambda)$  implies   $\widetilde{\psi} \in \mathrm{H}^2_{\mathrm{loc}}(\R^d)$ and
\beq\label{eqwidehat}
\widetilde{ (-\Delta +V) {\psi} }=   (-\Delta  + \widehat{V} ) \widetilde{\psi} .
\eeq

For periodic boundary condition,
we extend  $\vphi \in \mathrm{L}^2(\Lambda)$  and $V$ to  periodic functions $\widetilde{\vphi}$ and $\widehat{V}$ on $\R^d$ of period $(L_1,\ldots,L_d)$;  note  $\|\widehat{V}\|_\infty =\norm{V}_\infty \le K$.  Moreover, 
$\psi \in \cD(\Delta_\Lambda)$  implies   $\widetilde{\psi} \in \mathrm{H}^2_{\mathrm{loc}}(\R^d)$ and we have  \eq{eqwidehat}.
\smallskip

Let  $\btau=(\tau_1,\ldots,\tau_d)$ be given by 
\beq
\tau_j= \min \set{t\ge 2; L_j\in t\N_{\mathrm{odd}}}=\frac {L_j }{ 2\left\lfloor{\tfrac {L_j -2}4}\right\rfloor +1} , \quad j=1,2,\ldots,d.
\eeq
 It follows that ($L_j >12\sqrt{d} \ge 12$)
\beq
2\le \tau_j \le \frac {L_j} { 2\left(\tfrac {L_j-2}4  -1\right) +1}= \frac 2 {1-\frac 4 {L_j}}<\frac 2 {1-\frac 4 {12}}=3,
\eeq
so
\beq
\tau_\infty= \max_{j=1,\ldots,d} \tau_j <3.
\eeq
We let $\btau \Z^d= \prod_{j=1}^d \tau_j \Z$ and  $\Lambda\up{\btau} = \pa{a + \btau \Z^d}\cap \La$.
Then
\beq\label{Ladecomp}
\overline{\La}=\bigcup_{\kappa \in\Lambda\up{\btau}} \overline{\La_{\btau}(\kappa)}.
\eeq
We define $\cJ:\Lambda\up{\btau} \to\widehat{\widehat{\La}}$ in such a way that  $\La_1(\cJ(\kappa)) \subset   \La_{\btau} (\kappa)$ for all $\kappa\in \Lambda\up{\btau} $. This can always be done since $\tau_j\ge 2$ for $j=1,\ldots,d$; note that $\cJ$ is one to one.

Let  $Y \in \N_{\mathrm{odd}}$, $Y\le \frac {L_j} 6  <\frac {L_j} {2\tau_j}$ for $j=1,2,\ldots,d$. 
It follows that  for  all $ \vphi \in \mathrm{L}^2(\Lambda)$ we have (see  \cite[Subsection~5.2]{RV})
\beq\label{sumY}
\sum_{\kappa \in\Lambda\up{\btau}} \norm{\widetilde{\vphi}_{\La_{Y\btau}(\kappa)}}_2^2
\le (2\tau_\infty Y)^d \norm{\vphi_\La}_2^2 \le  (6Y)^d \norm{\vphi_\La}_2^2.
\eeq

We now fix   $\psi \in \cD(\Delta_\La) $.   Following Rojas-Molina and Veseli\' c, we  call a site $\kappa \in\Lambda\up{\btau}$ \emph{dominating} (for $\psi$) if
\beq\label{dom}
\norm{\psi_{\La_{\btau}(\kappa))}}_2^2 \ge\tfrac 1{2 (6Y)^{d}} \norm{\widetilde{\psi}_{\La_{Y\btau}(\kappa)}}_2^2.
\eeq
Letting  $\widehat{D}\subset\Lambda\up{\btau}$ denote the collection of dominating sites, Rojas-Molina and Veseli\' c   \cite[Subsection~5.2]{RV} observed that it follows from \eq{sumY},  \eq{dom}, and \eq{Ladecomp},  that
\beq\label{sumdom}
\sum_{\kappa \in \widehat{D}} \norm{\psi_{\La_{\btau}(\kappa))}}_2^2 \ge \tfrac 1 2 \norm{\psi_\La}_2^2.  
\eeq

We  define a map $J\colon  \Lambda\up{\btau} \to \Lambda\up{\btau}$ by
\beq \label{defJ}
 J(\kappa)=
\begin{cases}   \kappa+ 2\tau_1 \e_1 \qtx{if} \kappa+ 2\tau_1 \e_1\in \Lambda\up{\btau}\\
\kappa -  2\tau_1  \e_1 \qtx{if} \kappa +   2\tau_1  \e_1 \notin\Lambda\up{\btau}\end{cases}.
\eeq
Note that $J$ is well defined, and 
\beq\label{Jinv}
\# J^{-1} (\set{\kappa}) \le 2 \qtx{for all} \kappa \in \Lambda\up{\btau}.
\eeq
We have (see \eq{defRx0}) 
\begin{align}\label{Qest}
  {Q}(y_{\cJ(J(\kappa))},\La_{\btau}(\kappa)) & = \sup_{x \in \La_{\btau}(\kappa)} \abs{x- y_{\cJ(J(\kappa))}}\\ \notag &   \le  \sup_{x \in \La_{\btau}(\kappa)} \abs{x- \kappa} +\abs{\kappa- J(\kappa)}+
 \abs{J(\kappa)-y_{\cJ(J(\kappa))}} \\
 \notag & \le \tfrac {\tau_\infty}{2} \sqrt{d}  + 2 \tau_1 +  \tfrac {\tau_\infty}{2}  \sqrt{d} \le  3\tau_\infty \sqrt{d}  \le 9\sqrt{d}.
\end{align}
for all $k \in \Lambda\up{\btau}$.

For each $\kappa \in \Lambda\up{\btau}$  we will   apply Theorem~\ref{thmucp}   with $\Omega= \La_{Y\btau}(\kappa) $, $\Theta= \La_{\btau}(\kappa) $, and $x_0= y_{\cJ(J(\kappa))}$. We need to guarantee \eq{xR}, that is, 
\beq
B(y_{\cJ(J(\kappa))}, 6 {Q}(y_{\cJ(J(\kappa))},\La_{\btau}(\kappa)) +2)\subset \La_{Y\btau}(\kappa)  \qtx{for all} k \in \widehat{D}.
\eeq
It suffices, using \eq{Qest} and $\tau_j \ge 2$, to have
\beq
\abs{y_{\cJ(J(\kappa))}-\kappa } +9 \sqrt{d}  \le  \tfrac {\tau_\infty}{2} \sqrt{d}  + 2 \tau_1+9 \sqrt{d}\le \tfrac {33} 2\sqrt{d} \le \tfrac {2Y} 2 =Y.
\eeq
We thus choose
\begin{align}\label{Y36}
Y= \min\set{n \in   \N_{\mathrm{odd}}; n\ge  \tfrac {33} 2\sqrt{d} }\le   \tfrac {33} 2\sqrt{d} +2 \le 19\sqrt{d}.
\end{align}
Since we want  $Y\le \frac {L_j} 6  <\frac {L_j} {2\tau_j}$ for $j=1,2,\ldots,d$, we require 
$L_j \ge 114 \sqrt{d}$ for $j=1,2,\ldots,d$.

Applying  Theorem~\ref{thmucp}, for each $\kappa \in \Lambda\up{\btau}$  we get   
\begin{align} \label{UCPbound3}
 \delta^{m^{\prime}_d\pa{1 + K^{\frac 2 3}}}\norm{ {\psi}_{{\La_{\btau}(\kappa)}}}^2_2  \le   \norm{ {\psi}_{y_{\cJ(J(\kappa))},\delta}}^2_2 + \delta^2 \norm{{\widetilde{\zeta}_{\La_{Y\btau}(\kappa) }}}_2^2,
\end{align}
where  $\zeta= (-\Delta +V)\psi$ and  $m^{\prime}_d>0$ is a constant depending only on $d$.  Summing over $\kappa\in \widehat{D}$ and using \eq{sumdom}, \eq{Jinv},   \eq{sumY}, and \eq{Y36},  and we get
\begin{align}
\tfrac 1 2  \delta^{m^{\prime}_d\pa{1 + K^{\frac 2 3}}}\ \norm{\psi_\La}_2^2 &\le  
2\sum_{k \in\widehat{\widehat{\La}}}\norm{ {\psi}_{y_{k},\delta}}^2_2 +  (6Y)^d \delta^2 \norm{{{\zeta_\La}}}_2^2\\
& \le 2\sum_{k \in \widehat{\widehat{\La}}}\norm{ {\psi}_{y_{k},\delta}}^2_2 +  (114 \sqrt{d})^d \delta^2 \norm{{{\zeta_\La}}}_2^2,  \notag\end{align}
so \eq{UCPdelta} follows.
\end{proof}

\begin{proof}[Proof of  Theorem~\ref{thmUCPSP}] 

Given $E_{0}>0$, set $K =K(V,E_0)= 2\norm{V}_{\infty}+ E_{0} $, and   let $\gamma$ be given by \eq{defgamma}, where $M_d>0$ is the constant in Theorem~\ref{lemUCPeig}.  Let $I \subset (-\infty, E_{0}]$ be a closed interval with $\abs{I} \le 2\gamma $.  Since  $\sigma (H_\La)\subset [-\norm{V }_\infty,\infty)$ for any rectangle $\La$,  without loss of  generality we assume $I=[E-\gamma ,E + \gamma ]$ with $E \in [ -\norm{V }_\infty,E_{0}]$, so 
 \beq
\norm{V -E}_\infty \le \norm{V }_\infty + \max \set{E_0, \norm{V }_\infty}\le K.
\eeq
Moreover, for any 
rectangle $\La$  we have  
 \beq\label{IgammaI}
   \norm{{\pa{H_\La -E}\psi}}_{2}\le  \gamma \norm{\psi}_{2} \qtx{for all} \psi \in \Ran  {\Chi}_I(H_\La).
  \eeq

Let $\La$ be a rectangle as in Theorem~\ref{lemUCPeig} and $\psi \in \Ran  {\Chi}_I(H_\La)$. If $\psi$  is real-valued,  it follows from Theorem~\ref{lemUCPeig}, \eq{defgamma}, and \eq{IgammaI} that
\begin{align} \label{UCPdelta2}
 2\gamma^2\norm{\psi}_2^2   \le  
\sum_{k \in \widehat{\widehat{\La}}}\norm{ {\psi}_{y_{k},\delta}}^2_2 +  \gamma^2 \norm{\psi}_2^2, 
\end{align}
yielding \beq\label{UCPdelta248}
\gamma^2\norm{\psi}_2^2   \le \sum_{k \in\widehat{\widehat{\La}}}\norm{ {\psi}_{y_{k},\delta}}^2_2= \norm{{W\up{\La} \psi}}_2^2,
\eeq
where the equality follows from  \eq{defWthm}.  For arbitrary $\psi \in \Ran  {\Chi}_I(H_\La)$, we write
$\psi = \Re \psi + i\Im \psi$, and note that $\Re \psi, \Im \psi  \in \Ran  {\Chi}_I(H_\La)$, $\norm{\psi}_2^2= \norm{\Re \psi}_2^2 + \norm{\Im\psi}_2^2$, and, since $W\up{\La}$ is real-valued, $\norm{W\up{\La}\psi}_2^2= \norm{W\up{\La}\Re \psi}_2^2 + \norm{W\up{\La}\Im\psi}_2^2$.  Recalling     $\pa{W\up{\La}}^2=W\up{\La}$, we conclude that
\beq\label{UCPdelta2489}
\gamma^2 \scal{\psi,\psi}=\gamma^2\norm{\psi}_2^2   \le  \norm{{W\up{\La} \psi}}_2^2= \scal{\psi,W\up{\La} \psi}
\eeq
for all $\psi \in \Ran  {\Chi}_I(H_\La)$, proving \eq{chivchi}.  
\end{proof}

\end{document}